%% file: Second.tex
\documentclass[journal]{IEEEtran}
\usepackage{amsmath,amssymb}
\usepackage{subfigure}
\usepackage{graphicx,graphics,color,psfrag}
\usepackage{cite,balance}
\usepackage{caption}
\usepackage{framed}
\allowdisplaybreaks
\usepackage{algorithm}
\usepackage{algorithmic}
\usepackage{accents}
\usepackage{amsthm}
\usepackage{bm}
\usepackage{url}
\usepackage[english]{babel}
\usepackage{multirow}
\usepackage{enumerate}
\usepackage{cases}
\usepackage{stfloats}
\usepackage{dsfont}
\usepackage{color,soul}
\usepackage{amsfonts}
\usepackage{cite,graphicx,amsmath,amssymb}
\usepackage{subfigure}
\usepackage{fancyhdr}
\usepackage{hhline}
\usepackage{graphicx,graphics}
\usepackage{array,color}
\usepackage{mathtools}
\usepackage{amsmath}
\usepackage[T1]{fontenc}
\usepackage{float}  %设置图片浮动位置的宏包
\usepackage{subfigure} 
\usepackage{stfloats}
\usepackage{siunitx}

\include{header}

\begin{document}
\captionsetup[figure]{name={Fig.}}
\title{\huge Near-field Beam Training with Sparse DFT Codebook} 
\author{Cong Zhou, Chenyu Wu, Changsheng You,~\IEEEmembership{Member,~IEEE}, Jiasi Zhou, Shuo Shi,~\IEEEmembership{Member,~IEEE}
	% <-this % stops a space
	\thanks{Cong Zhou, Chenyu Wu and Shuo Shi are with the School of Electronic and Information Engineering, Harbin Institute of Technology, Harbin, 150001, China. (e-mail: zhoucong@stu.hit.edu.cn, wuchenyu@hit.edu.cn, crcss@hit.edu.cn).}
	\thanks{Changsheng You is with the Department of Electronic and Electrical Engineering, Southern University of Science and Technology, Shenzhen 518055, China. (e-mails: youcs@sustech.edu.cn).}
	\thanks{Jiasi Zhou is with the School of Medical Information and Engineering, Xuzhou Medical University, Xuzhou, 221004, China. (e-mails: jiasi{\_}zhou@xzhmu.edu.cn).}% <-this % stops a space
	%\thanks{Manuscript received April 19, 2021; revised August 16, 2021.} 
	\thanks{\emph{Corresponding author: Chenyu Wu and Changsheng You.}}
}

\maketitle
\begin{abstract}
	Extremely large-scale arrays (XL-arrays) have emerged as one promising technology to improve the spectral efficiency and spatial resolution in future sixth generation (6G) wireless systems.
	The drastic increase in the number of antennas renders the communication users more likely to be located in the near-field region, which requires a more accurate spherical (instead of planar) wavefront propagation modeling.
	However, this also inevitably incurs unaffordable beam training overhead when performing a two-dimensional (2D) beam-search in both the angular and range domains.
	To address this issue, we first introduce in this paper a new \emph{sparse} discrete Fourier transform (DFT) codebook, which exhibits the angular periodicity in the received beam pattern at the user.
	This thus motivates us to propose a three-phase beam training scheme.
	Specifically, in the first phase, we utilize the sparse DFT codebook for beam sweeping in an angular subspace and estimate candidate user angles according to the received beam pattern.
	Then, a central subarray is activated to scan specific candidate angles for resolving the issue of angular ambiguity for identifying the user angle.
	In the third phase, the polar-domain codebook is applied in the estimated angle to search the best effective user range.
	Finally, numerical results show that our proposed beam training scheme enabled by the sparse DFT codebook achieves 98.67\% beam training overhead reduction as compared to the exhaustive-search scheme, yet without compromising rate performance in the high signal-to-ratio (SNR) regime.
\end{abstract}

%------------------------------------------------------------------------------第一部分
\section{Introduction}
Extremely large-scale arrays/surfaces (XL-arrays/surfaces) have been envisioned as one of the key ingredients to drive the evolution of six generation (6G) wireless systems\cite{nf_mag,nf_bm,xlmimo,WuIntelligent2021}.
Specifically, XL-arrays/surfaces with a significant number of antennas can be deployed at the base station (BS) to achieve ultra-high spectral efficiency and spatial resolution, hence accommodating the escalating demands for new applications such as spanning metaverse and digital twin~\cite{liu_review, Khan2022DigitalTwin,Paul2024DigitalTwin}. The drastic increase in the number of antennas in high-frequency bands represents a qualitative paradigm shift in the electromagnetic (EM) propagation modeling, giving rise to the new \textit{near-field communications} \cite{zhang2023nearYou}. 

Particularly, different from the far-field EM propagation which is simply approximated by planar waves, the near-field channel modeling necessitates the use of more accurate spherical waves\cite{9723331,you2024next,an2024near}.  As such, near-field communications possess several unique properties in contrast to far-field communications. First, the spherical wavefront characteristic opens up the possibility of near-field \textit{beamfocusing}, for which the beam energy can be concentrated at a specific location/region rather than a spatial angle typically for far-field beamforming\cite{beamfocusing,10068140,nepa2017near}. The beam-focusing capability of near-field communications enables XL-array to flexibly form highly directional beams in both the angle and range domain, and hence can be leveraged in various applications to improve the system performance, such as mitigating the inter-user interference, improving the accuracy of sensing and localization\cite{cong2023near,chen2024integrated,liuyuanwei2024near}, and enhancing the charging efficiency of wireless power transfer\cite{wpt}. Second, the rank of line-of-sight (LoS) channels for near-field multiple-input multiple-output (MIMO) communication systems can be larger than one, hence enhancing the spatial multiplexing gains~\cite{10117500}. 
In this paper, we propose an efficient near-field beam training scheme with an sparse discrete Fourier transform (DFT) codebook by using the sparse antenna activation method, which significantly reduces the beam training overhead.
\subsection{Related Works}
\subsubsection{Near-Field Wireless Systems} Near-field communications bring new  opportunities and challenges, which has motivated upsurging research interest. For example, in \cite{beamfocusing}, the authors studied the beamfocusing design based on the fully-digital architectures, hybrid phase shifter-based precoders, and dynamic metasurface antenna architecture for XL-MIMO arrays. It is shown that near-field beamfocusing provides new degree of freedom (DoF) to mitigate the interference among users, even when they are located at the same spatial angle. In addition, a new concept of location division multiple access (LDMA) was proposed in \cite{ldma}, which exploits orthogonality of the near-field beamfocusing vectors in the range domain to serve different users at the same angle. The authors in \cite{kangda} developed a framework for analyzing and designing XL-MIMO systems with spatial non-stationarity.
It was revealed that the performance of the proposed framework approaches that of the conventional full-antenna array based designs albeit with lower complexity. In addition, a holographic metasurface antennas (HMAs) based multi-user system was investigated in \cite{nf_holo}, where the digital transmit precoder and the analog HMA weighting matrix were jointly optimized to minimize the transmit power. The authors in \cite{yunpu_swipt} considered a simultaneous wireless information and power transfer (SWIPT) system, where energy harvesting (EH) and information decoding receivers are located in the near- and far-field regions of the XL-array, respectively. 
The beam scheduling and power allocation were jointly optimized to maximize the weighted-sum power harvested at EH receivers~\cite{yunpu_swipt}.
Moreover, a directional modulation system was developed for near-field physical layer security systems\cite{nf_pls2}. Specifically, a fully analog precoding algorithm along with artificial noise and power allocation was proposed to realize secure transmission in both the angular and range domains. Near-field  sensing with XL-array was studied in \cite{nf_sensing}, where the closed-form expressions of the Cramér-Rao Bounds for both the angle and range estimations are derived.
\subsubsection{Near-field Beam Training}
In high-frequency bands, direct channel estimation methods may not be very effective due to severe path-loss and signal misalignment. As such, beam training is efficient in establishing initial links with high signal-to-noise ratio (SNR) for data transmission and channel state information (CSI) acquisition \cite{zhang2022dual}. However, due to the spherical wave propagation, near-field beam training is more challenging compared with its far-field counterpart as it requires a joint beam search over both the angular and range domain. Particularly, conventional far-field beam training will suffer from significant performance loss in the near-field region due to the so-called \textit{energy-spread effect}, for which the energy of a far-field beamformer is no longer steered towards one angle, but spread in multiple angles. 
Hence, the conventional far-field beam training method cannot be directly applied to the near-field beam training. To address this issue, the authors in \cite{nf_exhaustive} proposed a new codebook design in the polar domain, for which the angular domain is uniformly sampled whereas the range domain is non-uniformly sampled. One can simply invoke this codebook for the exhaustive-search based near-field beam training. However, the overhead of this scheme is the product of the number of antennas and range samples, which is prohibitively high for the implementation of XL-arrays. 
To reduce the overhead, the authors in \cite{two_phase} proposed a novel two-phase near-field beam training method. This method leveraged a key observation that the true user angle approximately lies in the middle of an angular support with high received SNRs. Based on this observation, one can first estimate the user angle by using far-field DFT beams, and then estimate the user range with the polar-domain codebook in~\cite{nf_exhaustive}. Besides, deep learning techniques have also been exploited  in \cite{dl} and \cite{10163797} to reduce the near-field beam training overhead, where deep neural networks (DNN) are trained based on conventional far-field codebooks and near-field codebook, respectively. However, the training overhead of these methods scales linearly with the number of antennas, which is still unaffordable for communication systems. This issue motivates the design of efficient hierarchical beam training schemes for near-field communications to reduce the training overhead to the logarithmic order, e.g., \cite{twostage,chirp_hier,dai_hier}.
However, hierarchical beam training schemes suffer from several inherent drawbacks.
First, hierarchical methods usually require frequent feedback and neglect underlying transmission delays.
Moreover, there exists the error propagation issue due to the progressive beam search \cite{noh2015multi}.
%Beam training and tracking methods are jointly investigated for XL-MIMO systems with partially-connected hybrid precoding structures in \cite{10288339}. 

\subsection{Motivations and Contributions}
The existing works on near-field beam training inevitably incur unacceptable overhead, while the off-grid channel estimation methods face highly computational complexity such as high dimensional matrix inversion.
Moreover, the beam training overhead of these existing works scales linearly with the number of antennas apart from hierarchical beam training schemes, which suffers from several inherent drawbacks such as error propagation and user feed-back overhead~\cite{noh2015multi}.
Motivated by the above, this paper explores a new sparse DFT codebook and a three-phase beam training scheme to reduce the near-field beam training overhead, which scales with the square root of the number of antennas.
The main contributions are summarized as follows.
\begin{itemize}
	\item First, we propose a novel sparse DFT codebook by sparsely activating the XL-array antennas and constructing the reduced DFT codebook with the equivalent sparse linear array (SLA).
	Specifically, the sparse DFT codebook consists of sparse far-field channel response vectors, which is sampled from far-field channel response vectors.
	Then, we characterize the received beam pattern at the near-field user when the sparse DFT codebook is used for beam sweeping. Interestingly, it is shown that the received beam pattern exhibits periodicity in the angular domain, while there still exists the energy-spread effect.
%	Moreover, we comprehensively analyze the received beam pattern and mathematically 
%	elaborate the energy spread phenomenon.
	Then, we show that the user angle information is contained in a period of the received beam pattern at the user and can be estimated via a defined \emph{angular support}.
	\item Second, we propose a novel three-phase beam training scheme based on the  sparse DFT codebook.
	Specifically, in the first phase, we utilize a small number of the sparse DFT codewords to sweep an angular subspace and estimate one candidate user angle according to the middle of the defined angular support.
	Then, in the second phase, we activate a central subarray to resolve the angular ambiguity by virtue of the periodicity of the received beam pattern at the user.
	Subsequently, the polar-domain codebook is utilized to search the best user range in the estimated user angle.
	\item Finally, extensive numerical results are presented to demonstrate the effectiveness of our proposed beam training scheme enabled by the sparse DFT codebook.
	It is shown that the proposed three-phase beam training scheme can achieve nearly the same performance with the exhaustive-search beam training method in the high-SNR regime, while reducing more than 98\% of the training overhead.
	In the low-SNR regime, the proposed scheme suffers from slight performance loss, while the effective rate still significantly exceeds all benchmark schemes due to lower beam training overhead.
\end{itemize}
\subsection{Organization and Notations}
The remainder of this paper is organized as follows.
System model is presented in Sections \ref{Sec:System model}.
In Section \ref{Sec:Bench}, we introduce several benchmarks.
Section \ref{section:Codebook Design} provides comprehensive analysis of the received beam pattern at the user with the sparse DFT codebook.
Then, in Section \ref{Sec:scheme2}, the proposed three-phase beam training scheme is elaborated. 
Finally, numerical results are provided in section \ref{Sec:numericalResults} to demonstrate the effectiveness of the proposed beam training scheme followed by the conclusions made in Section \ref{Sec:Conclusion}. 

\textit{Notations}: Vectors and matrices are respectively denoted by lower-case and upper-case boldface letters.
The symbol $ \left|\cdot\right| $ represents the absolute value, while $ \left\lVert \cdot\right\lVert $ denotes the $ l_{2} $ norm.
Moreover, we use $ (\cdot)^{H} $ to denote the conjugate transpose operation.
Finally, the Hadamard product is represented by $ \odot $. The key symbols used in this paper are listed in Table \ref{table1}.
\begin{table*}[htb]
	\renewcommand{\arraystretch}{1.5}
	\setlength{\tabcolsep}{6pt}
	\centering
	\caption{List of main symbols and their physical meanings.}
	\label{table1}
	\begin{tabular}{|c|l|c|l|}
		\hline
		$N$  & Number of BS antennas  &  $U $ & Antenna activation interval \\ \hline
		$D$ & Array aperture size  & $M $ & Number of antennas of the activated subarray \\ \hline
		$\lambda$  & Carrier wavelength   &$z$ & AWGN \\ \hline
		$d_0$  & Antenna spacing  &$Q$  & Number of antennas of the activated SLA \\ \hline
		$\mathbf{h}^H_{\rm near}$ & Near-field channel & $\kappa$ & Rician factor
		\\ \hline
		$L$  & Number of channel paths & ${{{P}}_{\rm tot}}$ & Transmit SNR    \\ \hline
		$\theta_{0}$ & BS center-user spatial angle  & ${{\mathcal{W}}}_{\rm{DFT}}$ & Sparse DFT codebook \\ \hline
		$r_{0}$  & BS center-user range & $\beta$ & Channel gain \\ \hline
		$\mathbf{b}\left(r_{0}, \theta_{0} \right) $ & Near-field channel steering vector & $\Delta $ & Difference of spatial angles \\ \hline
		${\mathbf{w}}$ & Beamforming vector & $f({r_0},{\theta _0};\theta )$  & Received beam pattern \\ \hline
		${{\mathcal{V}}}_{\rm{Sub}}$ & DFT codebook for the activated subarray & $\bar{\mathcal{X}}_{\rm{Pol}}$  & Polar-domain codebook \\ \hline
	\end{tabular}
\end{table*}
\section{System Model} 
\label{Sec:System model}
We consider a single-user XL-array downlink communication system, where the BS is equipped with a dense uniform linear array (ULA) with $ N $  antennas. 
In this section, the near-field channel and signal model for the ULA are introduced. 
%\underline{\bf Near-field channel model:}
\subsection{ Near-field Channel Model}
We assume that the dense ULA is situated at the $y$-axis and centered at the origin. Specifically, each antenna of XL-array is located at ($0, nd_{0}$), where $ n \in \mathcal{N} \triangleq \{0,\pm 1,\cdots, \pm \frac{N-1}{2}\} $ and $ d_{0}$ respectively denote the antenna index and inter spacing.
For the dense ULA, we have $ d_{0} = \frac{\lambda}{2}$, where $ \lambda $ represents the carrier wavelength.
Moreover, the single user is assumed to be located in the Fresnel near-field region of the XL-array where the BS-user range $ r_{0} $ is larger than the Fresnel distance $Z_{\rm F}=\max{\{ d_{R} , 1.2D \} } $ and smaller than the Rayleigh distance $Z_{\rm R}=\frac{2D^2}{\lambda}$ with $ D = (N-1)d_{0} $ denoting the array aperture.
Moreover, $ d_{R} $ is proven to be several wavelengths in \cite{ouyang2024impact} and the Fresnel distance can be simplified by $ Z_{\rm F} = 1.2D $.
Hence, the line-of-sight (LoS) channel follows the uniform spherical wave (USW) model \cite{lu2023tutorial}.
For example, when $N=257$ and $f = 30$ GHz, the Rayleigh distance is approximately $328$ m, which makes the user more likely to be located in the near-field region.
Then the general  multi-path channel from the XL-array to the user can be modeled as \cite{yunpu_swipt}
\begin{equation}
\label{Eq:nf-channel}
\mathbf{h}^H_{\rm near} = \sqrt{N}\beta \mathbf{b}^{H}(r_{0}, \theta_{0}) + \sum_{\ell=1}^{L} \sqrt{\frac{N}{L}}\beta_{\ell} \mathbf{b}^{H}(\bar{r}_{\ell}, \bar{\theta}_{\ell}),
\end{equation}
which includes one LoS path $ \mathbf{h}_{\rm LoS}^H $ and $L$ non-LoS (NLoS) paths.
Herein, the parameters $ r_0 $ ($ {\bar r_\ell} $) and $ \theta_{0} $ ($ \bar{\theta}_{\ell} $) represent the range and spatial angle of the LoS ($ \ell $-th NLoS) signal path.
Moreover, the parameters $ \beta $ and $ \beta_{\ell} $ denote the LoS path and $ \ell $-th NLoS path gain, respectively.
Mathematically,  $ \beta $ can be modeled as \cite{zhang2020capacity}
\begin{equation}
	\beta_{}=\sqrt{\frac{\kappa}{\kappa+1}} \frac{\sqrt{\beta_0}}{r_{{0}}} e^{-\frac{\jmath 2 \pi r_{{0}}}{\lambda}}, 
\end{equation}
where $ \kappa $ and $ \beta_0 $ represent the Rician factor and reference channel
gain at a range of 1 m, respectively.
\begin{figure}
	\centering
	\includegraphics[width=7.5cm]{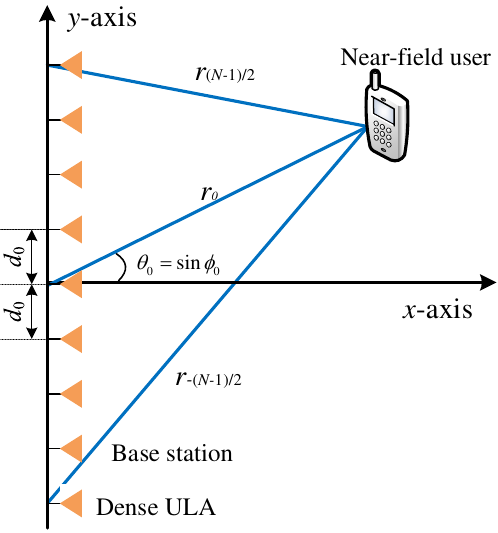}
	\caption{A narrow-band Near-field XL-array communication system.}
	\label{fig:systemModel}
\end{figure}

In this paper, we mainly consider the near-field communication scenarios in high-frequency bands such as millimeter-wave (mmWave) and even terahertz (THz). 
In these scenarios, the NLoS channel paths exhibit negligible power owing to the severe path-loss and shadowing effects \cite{nf_ris}. 
Therefore, we only consider the LoS channel and the BS-user channel can be approximated as $\mathbf{h}^H_{\rm near} \approx \mathbf{h}_{\rm LoS} $ \footnote{{The case where comparable multi-path components exist is more complicated and discussions are provided in the Section \ref{sec:discussion}. We will evaluate the Rician fading channel in the simulation results.}}.
Based on USW model, the near-field LoS channel from BS$\to$user can be modeled as~\cite{two_phase}
\begin{equation}\label{Eq:nf-model}
	\mathbf{h}^H_{\rm near}\approx \sqrt{N}\beta \mathbf{b}^{H}(r_{0}, \theta_{0}),
\end{equation}
where $\mathbf{b}^{H}(r_{0},\theta_{0})$ denotes the near-field channel steering vector, defined as \cite{yunpu2} 
\begin{equation}\label{near_steering}
	\left[\mathbf{b}^H\left(r_{0},\theta_{0}\right) \right]_{n} = \frac{1}{\sqrt{N}}e^{-\frac{\jmath 2 \pi r_{n}}{\lambda}}, \forall n\in \mathcal{N},
\end{equation}
with $r_{n}=\sqrt{r_{0}^2+ {n}^2d_0^2-2r_{0}\theta_{0} nd_0}$ denoting range between the $ n $-th antenna and the user. 
Moreover, $\theta_{0} = \cos \phi_0 \in[-1,1]$ represents the spatial angle at the BS, with $ \phi_0 $ denoting the physical angle-of-departure (AoD) from the BS center to the user.
Further, by means of Fresnel approximation, $ r_n $ can be approximated as
\begin{equation}
\label{eq:Fresnel approximation}
	r_{n} \approx r_{0}-nd_0 \theta_{0}+\frac{n^2 d_0^2 (1- \theta_{0}^2)}{2 r_{0}},
\end{equation}
which is shown to be accurate in \cite{liu_review}.
\subsection{Signal Model} 
%\underline{\bf Signal model:}
Let $x\in \mathbb{C}$ denote the transmitted symbol by the BS with unit power and $\bar{\mathbf{w}}\in \mathbb{C}^{N \times 1}$ represent the beamforming vector~\cite{8030501}. Then the received signal at the user is given by
\begin{align}\label{Eq:general_Sig}
	y(\bar{\mathbf{w}})
	%    \mathbf{h}^H_{\rm near}\mathbf{v}x+z
	=\sqrt{N}\beta \sqrt{P_{\rm tot}}\mathbf{b}^{H}(r_{0},\theta_{0})\bar{\mathbf{w}}x+z,
\end{align}
where $z$ is the received additive white Gaussian noise (AWGN) and $z \sim  \mathcal{C N}\left(0, \sigma^2\right) $. Moreover, $ P_{\rm tot} $ denotes the total transmit power of the BS.
Then, the achievable rate in bits/second/hertz (bps/Hz) is given by
\begin{equation}
	R=\log_{2}\l(1+\frac{P_{\rm tot} N|\beta |^2|\mathbf{b}^{H}(r_{0},\theta_{0})\bar{\mathbf{w}}|^2}{\sigma^2}\r).
\end{equation}

\section{Benchmark Beam Training Schemes}
\label{Sec:Bench}
In this section, two benchmark beam training schemes and their drawbacks are presented.
%In this section, we introduce two benchmark schemes for near-field beam training and discuss their main limitations.
% This section shows two benchmark near-field beam training schemes and their limitations.
\subsection{2D Exhaustive-Search Beam Training Method}
\label{Sec:Exhaustive Search} 
The authors in \cite{nf_exhaustive} proposed a \emph{polar-domain} codebook, each steering a focusing beam to a specific location.
In particular, the angular domain is uniformly sampled, while the range domain is \emph{non-uniformly} sampled.
Specifically, the polar-domain codebook is given by 
\begin{equation}
\label{eq:cui polar-domain codebook}
\bar{\mathcal{W}}_{\rm{Pol}} = \{\bar{\mathcal{W}}_1,\cdots,\bar{\mathcal{W}}_{\bar n},\cdots,\bar{\mathcal{W}}_{N} \},
\end{equation}
where $ \bar{\mathcal{W}}_{\bar n} = \{ \bar{\mathbf{w}}_{\bar{n},1} \cdots,\bar{\mathbf{w}}_{\bar{n},v},\cdots \bar{\mathbf{w}}_{\bar{n},V}\}$ denotes the sub-codebook steering $ V $ beams towards the angles $ \theta_{\bar{n}} = \frac{2 {\bar n}-N+1}{N}, \forall {\bar n} \in \bar{\mathcal{N}} \triangleq \{1,2, \cdots, N\} $.
Mathematically, we have $ \bar{\mathbf{w}}_{{\bar n},v}=  \mathbf{b}\left(r_{{\bar n}, v}, \theta_{\bar n}\right)$ where $ r_{{\bar n},v}=\frac{1}{v} \alpha_{\Delta}\left(1-\theta_{{\bar n}}^2\right), \quad \forall v\in\mathcal{V}\triangleq\{1,2,3, \cdots V\} $
with $\alpha_{\Delta}\triangleq\frac{N^2 d_0^2}{2 \lambda \beta_{\Delta}^2}$.
Moreover, $\beta_{\Delta}$ is a constant corresponding to the quantization loss in the range domain \cite{nf_exhaustive}.
Given the polar-domain codebook $\bar{\mathcal{W}}_{\rm{Pol}}$, a \emph{two-dimensional exhaustive search} method can be directly applied in both the angular and range domains to search the best codeword, which yields the maximum received SNR at the user.
The beam training overhead of this exhaustive-search beam training method is $T^{\rm{(ex)}}=NV$, which is proportional to the product of the number of antennas and range samples.
When the number of antenna is large, the beam training overhead is unaffordable.

\subsection{Two-Phase Near-field Beam Training}
\label{Sec:Two phase}
To further reduce the beam training overhead, the authors in \cite{two_phase} proposed a two-phase near-field beam training method, which explored the so-called energy-spread phenomenon.
Specifically, when the far-field DFT codebook is used for the angular sweeping in the near-field, it is observed that the user angle approximated lies in the middle of an angular support region, which is to be defined in Section \ref{section:Codebook Design}.
Mathematically, the conventional DFT codebook is given by
\begin{align}\label{Eq:farDFT}
\bar{\mathbf{w}}_{\bar n}&=\mathbf{a}(\theta_{\bar n})\triangleq
\frac{1}{\sqrt{N}}\left[1,e^{-\jmath \pi\theta_{\bar n}},\cdots,e^{-\jmath \pi(N-1)\theta_{\bar n}}\right],
\end{align}
where $ \theta_{\bar{n}} = \frac{2 {\bar n}-N+1}{N}, \forall {\bar n} \in \bar{\mathcal{N}}$.
Given this observation, they explore the conventional DFT codebook in the first phase to perform beam sweeping, which estimates the user angle information. 
Then, given the candidate user angle, the polar-domain codebook in \eqref{eq:cui polar-domain codebook} is used to search the best user range in the second phase.
The beam training overhead of the two-phase beam training method is $T^{(\rm 2P)}=N+KV$ with $K$ representing the number of candidate user angles.
Although this method significantly reduces the beam training overhead of the exhaustive-search method, the overhead of the two-phase beam training method is still proportional to the number of antennas, which is prohibitively high as $N$ is sufficiently large.

To address the above issues, we propose a new near-field beam training method using a proposed sparse DFT codebook, which is equivalent to sparsely activating the dense ULA equipped by the BS, yielding extremely lower overhead as compared with various benchmark schemes.

\section{Received Beam Pattern of the Sparse DFT Codebook}
\label{section:Codebook Design}
In this section, we first introduce the sparse DFT codebook and then analyze its received beam pattern.
Moreover, we design a periodical beam training codebook to reduce overhead and propose to activate a central subarray for resolving the angular ambiguity.
%Hence, we propose a new multi-beam training scheme, leveraging the DFT codebook to further reduce the overheads.
The main definitions in this section are given as follows, including the received beam pattern and angular support.
\begin{definition}
	\emph{Given a fixed near-field user located at $ (r_{0}, \theta_{0}) $ and an arbitrary far-field beamforming vector $ \bar{\mathbf{w}} = \mathbf{a}(\theta) $ steering the beam towards the angle $ \theta $, the received beam pattern at the user is defined as
		\begin{equation}
			\label{received beam pattern in general}
			f( r_{0}, \theta_{0};\mathbf{\theta} )\triangleq|\mathbf{b}^H(r_{0},\theta_{0}) \mathbf{a}(\theta)|, \forall \theta.
	\end{equation} }
\end{definition}

\begin{definition}\label{De:angularSupport}
	\emph{Given a near-field channel response vector $\mathbf{b}^H(r_{0},\theta_{0})$ and a far-field beamforming vector $ \bar{\mathbf{w}} = \mathbf{a}(\theta) $ with $ \theta \in \mathcal{L} $, the $3$-dB \emph{angular support}  $\mathcal{A}_{\mu}^{\mathcal{L}}(r_{0}, \theta_{0})$ in the region $ {\mathcal{L}} $ is defined by \cite{two_phase}
		\begin{align}\label{Eq:AdB}
		\mathcal{A}^{\mathcal{L}}_{\mu}(r_{0},\theta_{0})\!\!=\!\!
		\left\{\theta\mid f\left(r_{0},\theta_{0}, \theta\right)\!>\!\kappa \max _{\theta \in \mathcal{L}} f\left(r_{0}, \theta_{0},\theta\right)\right\},
		\end{align}
		where $ \kappa = 10^{\mu/10}$. Moreover, let $\theta_{\rm left}$ and $\theta_{\rm right}$ be the smallest and largest angle in $\mathcal{A}^{\mathcal{L}}_{\mu}(r_{0}, \theta_{0})$. Then, its \emph{angular support width} is defined as 
		\begin{equation}
		\Gamma_{\mu}^{\mathcal{L}}(r_{0},\theta_{0})= \theta_{\rm right}-\theta_{\rm left}.
		\end{equation}}
\end{definition}
\begin{figure}
	\centering
	\includegraphics[width=7.5cm]{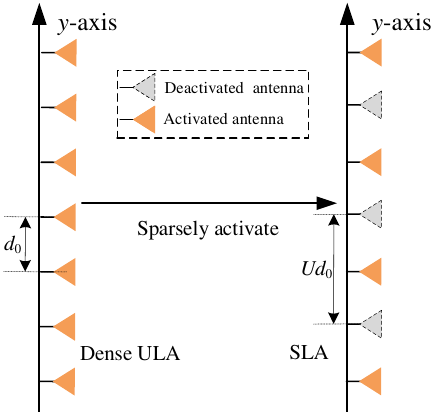}
	\caption{Illustration of sparse linear array.}
	\label{fig:sparse activation}
\end{figure}
\subsection{Sparse DFT Codebook}
Each column of the DFT Codebook in \eqref{Eq:farDFT} is a far-field channel response vector.
For the sparse DFT codebook, we sample each column of the original DFT codebook using an interval of $ U $, while the positions that are not sampled are padded with zeros. As such, $ \bar n $-th column of the sparse DFT Codebook $ \tilde{\mathbf{w}}_{\bar n} \in \mathbb{C}^{N \times 1} $ is given by $ \tilde{\mathbf{w}}_{\bar n} = \bar{\mathbf{w}}_{\bar n} \odot \mathbf{n}(U) $, where
\begin{equation}
\label{Eq:samped vector}
\mathbf{n}^{H}(U) = [1,\underbrace{0,\cdots,0,}_{U-1}1,\cdots,1]
\end{equation}
denotes the sampling vector.
It can be verified that the number of non-zero elements in $ \tilde{\mathbf{w}}_{\bar n} $ is $ Q =  \frac{N-1}{U}+1 $ (assuming $ Q $ is an integer for convenience).
Moreover, we rearrange the non-zero elements in $ \tilde{\mathbf{w}}_{\bar n} $ into a new vector $ {\mathbf{w}_{\bar n}} \in \mathbb{C}^{Q \times 1} $, which is given by
{\small
\begin{equation}
	\label{eq:SLA steering vector}
	{\mathbf{w}}_{\bar n} = \mathbf{a}_{\rm SLA}(\theta_{\bar n}) = \frac{1}{\sqrt{Q}}\bigg[1,e^{\jmath {\pi}U\theta_{\bar n}}, \cdots,e^{\jmath{\pi}(Q-1)U\theta_{\bar n}}\bigg]^H\!\!,
\end{equation}}
referred to as the \emph{sparse far-field beamforming vector}.

It is noteworthy that this sampling method is equivalent to sparsely activating a number of antennas of the the XL-array with an interval of $ U $, which effectively transforms the equipped ULA into an SLA with an inter-element spacing of $ Ud_0 $, as illustrated in~Fig.~\ref{fig:sparse activation}.
Furthermore, for brevity, we denote the channel response vector of the effective SLA as
\begin{equation}
	\left[\mathbf{b}^H_{\rm SLA}\left(r_0,\theta_0\right) \right]_{q} = \frac{1}{\sqrt{Q}}e^{-\frac{\jmath 2 \pi r_{q}}{\lambda}}, \forall q\in \mathcal{Q},
\end{equation}
where $ \mathcal{Q} \triangleq \{ 0, \pm 1, \cdots, \pm \frac{Q-1}{2} \} $ denotes the set of the SLA antenna index and $ r_{q} = \sqrt{r_0^{2} -2 q  U d_0 r_0 \theta_0 + (q U d_0)^2} $ represents the range between the user and $ q $-th antenna of the SLA.
Similar to \eqref{eq:Fresnel approximation}, $ r_{q} $ can be approximated as $ r_{q} \approx r-qUd_0 \theta_0+\frac{q^2 (Ud_0)^2 (1- \theta_0^2)}{2 r_0} $ with Fresnel approximation.
Then, the received signal with beamforming vector $ \mathbf{w}_{\bar n} $ at the user is rewritten as
\begin{align}\label{Eq:reritten received signal}
	y(\mathbf{w}_{\bar n})
	%    \mathbf{h}^H_{\rm near}\mathbf{v}x+z
	=\sqrt{Q}\beta \sqrt{P_{\rm tot}}\mathbf{b}^{H}_{\rm SLA}(r_{0},\theta_{0})\mathbf{w}_{\bar n}x+z.
\end{align}
Moreover, the received beam pattern at the user in \eqref{received beam pattern in general} can be rewritten as 
\begin{equation}
	\begin{aligned}
		f( r_{0}, \theta_{0};\theta_{\bar n})&=|\mathbf{b}^H(r_{0},\theta_{0}) {\bar{\mathbf{w}}_{\bar n}}| = |\mathbf{b}^H_{\rm SLA}(r_{0},\theta_{0}) {\mathbf{w}}_{\bar n}|\\
		&=|\mathbf{b}_{\rm SLA}^H(r_{0},\theta_{0}) \mathbf{a}_{\rm SLA}(\theta_{\bar n})|.
	\end{aligned}
\end{equation}

\subsection{Near-field Received Beam Pattern}
\label{Sec:DFTbeamPattern}
To obtain more insights, we first characterize the received beam pattern at the user with sparse far-field beamforming vectors spanning in the continuous spatial angles.
Let $ \mathbf{w}(\theta) = \mathbf{a}_{\rm SLA}(\theta)$ represent a sparse far-field beamforming vector, for which $ \theta = \theta_{\bar n}, \forall {\bar n} \in {\bar {\mathcal{N}}} $ is the discrete sampled angle.
We first characterize the received beam pattern of the sparse far-field beamforming vector
$ {\mathbf{w}} $ as follows.
\begin{lemma}
	\emph{For a sampled beamforming vector $ {\mathbf{w}} $ parameterized by $\{Q, U\}$, the received beam pattern at the user is given by 
		\begin{align}
		\label{eq:GeneralMixedField}
		&f \left( r_{0}, \theta_{0}; \theta \right) \nn\\
		\!\!\!\!\!\!\overset{(a_1)}{\approx}& \frac{1}{Q}\!\! \left| \sum_{q\in \mathcal{Q}}\!\exp{\l(\underbrace{\jmath  {\pi q  U \Delta}}_{B_1}+\underbrace{\jmath {\frac{\pi}{\lambda} q^2 ( U d_{0})^2\frac{1-\theta_{0}^2}{r_{0}} }}_{B_2}\r)} \!\right|\!\\
		\triangleq &\hat{f} ( r_{0}, \theta_{0};\theta),\nn
		\end{align}\noindent where $ \Delta\triangleq \theta-\theta_{\rm 0}$ and $(a_1)$ is due to the Fresnel approximation and shown to be accurate in \cite{kosasih2023finite}. 
	}
\end{lemma}

\begin{proposition}[The periodicity of $ \theta $]
	\label{pro:DFT periodicity}	{\rm 
	$ f \left( r_{0}, \theta_{0}; \theta \right) $ is a periodic function of $ \theta $ with the period of $ \frac{2}{U} $. Mathematically,  we have}
	\begin{equation}
		f \left( r_{0}, \theta_{0}; \theta \right) = f \left( r_{0}, \theta_{0}; \theta+\frac{2k}{U} \right), \forall k \in \mathbb{Z}.
	\end{equation}
\end{proposition} 
\begin{proof}
	For an arbitrary integer $ k $ and $ {\bf{a}}(\theta + \frac{2k}{U}) $, we have
	\begin{equation*}
	\begin{aligned}
	{[ {\bf{a}}(\theta + {2k}/{U}) ]_q} &= \exp (\jmath\pi qU(\theta + {2k}/{U}))\\
	&= \exp (\jmath\pi qU\theta )\exp (\jmath 2\pi kqU)\\
	&= \exp (\jmath\pi qU\theta ) = {[{\bf{a}}(\theta )]_q}.
	\end{aligned}	
	\end{equation*}
	{\rm Hence, we have $ f  \left( r_{0}, \theta_{0}; \theta \right) =  f \left( r_{0}, \theta_{0}; \theta+k\frac{2}{U} \right)$ and thus complete the proof.}
\end{proof}
\begin{figure}
	\centering
	\includegraphics[width=\columnwidth]{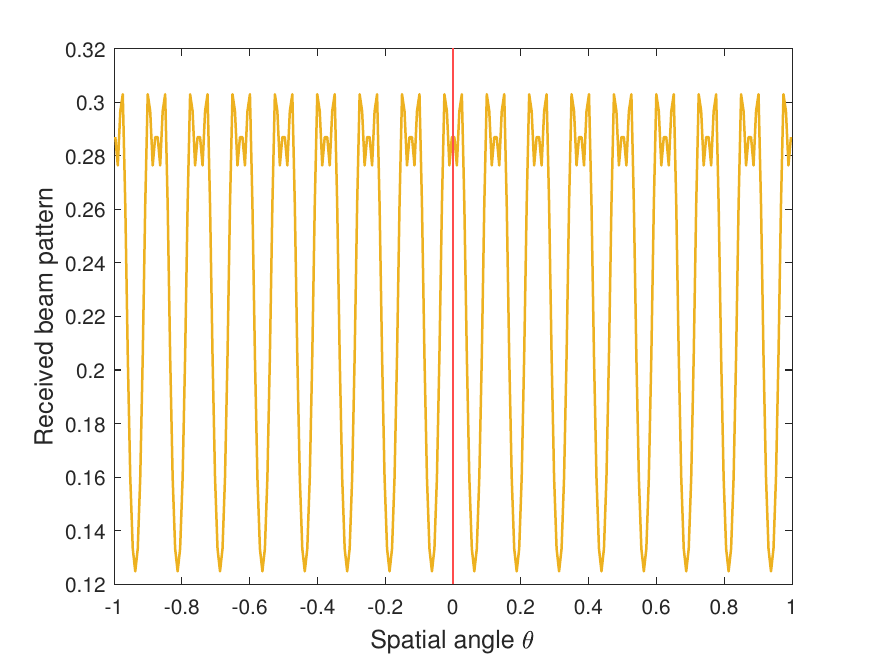}
	\caption{Received beam pattern, where $N=257$, $ U = 16 $ and $f=30$ GHz. The actual user angle is $ \theta_0 = 0 $ marked by the red line.}
	\label{fig:mixedFieldPeriod}
\end{figure}

In Fig. \ref{fig:mixedFieldPeriod}, we plot the received beam pattern versus the spatial angle $ \theta $, where the periodicity of $ \theta $ is exhibited.
Moreover, it is observed that the user angle is located in the middle of the angular support $\mathcal{A}^{\mathcal{L}_0}_{\mu}(r_{0}, \theta_{0})$, where $ \mathcal{L}_0 = [\theta_0-1/U,\theta_0+1/U) $ (a period $ \frac{2}{U} $).
Then, in the following, we prove that this observation holds for arbitrary user location $ (r_{0}, \theta_{0}) $.

Due to the periodicity of $ \theta $, we only need to focus on the region $ \mathcal{L}_0 $, i.e., $ \Delta \in  [-1/U,1/U)$, which inexplicitly contains useful user angle information.
%Then, we introduce the periodical energy spread phenomenon as shown in Fig. \ref{fig:mixedFieldPeriod} due to the periodicity involved in the sampled beamforming vector $ \mathbf{w} $.
Given $ \Delta \in  [-1/U,1/U)$, \eqref{eq:GeneralMixedField} can be approximated as follows.
\begin{lemma}\label{lemma:LSAgeneral}\emph{
		When $ \Delta \in  [-1/U,1/U)$, the received beam pattern $ \hat{f}( r_{0}, \theta_{0}; \theta)$ in \eqref{eq:GeneralMixedField} can be approximated as
		\begin{equation}\label{MixedFieldlosedForm}
		\begin{aligned}
			\hat{f}  \left( r_{0}, \theta_{0}; \theta \right) \approx \left| G(\beta_{1},  \beta_{2})\right|.
		\end{aligned}		
		\end{equation}
		Specifically, we have 
		\begin{equation}
			G(\beta_{1},  \beta_{2}) \triangleq \frac{\widehat{C}(\beta_{1},\beta_{2}) +  \jmath(\widehat{S}(\beta_{1},\beta_{2})}{2\beta_2},
		\end{equation}
		where $ \widehat{C}(\beta_{1},\beta_{2}) \triangleq {C}(\beta_{1}+\beta_{2}) - C(\beta_{1}-\beta_{2})$ {\rm and} $ \widehat{S}(\beta_{1},\beta_{2}) \triangleq S(\beta_{1}+\beta_{2}) - S(\beta_{1}-\beta_{2}) $. Further, $ C(x) $ and $ S(x) $ are Fresnel integrals, which are given by
		$$ C(x) = \int_{0}^{x} \cos(\frac{\pi}{2}t^2 ){\rm d}t,  S(x) = \int_{0}^{x} \sin(\frac{\pi}{2}t^2 ){\rm d}t. $$ Moreover, $ \beta_1 $ and $ \beta_2 $ are given by
		\begin{equation}
		\label{beta}
		\beta_{1} = \Delta\sqrt{\frac{r_{0}}{d_0(1-\theta_{0}^2)}},
		~~
		\beta_{2} = \frac{QU}{2}\sqrt{\frac{d_0(1-\theta_{0}^2)}{r_{0}}}.
		\end{equation}}
\end{lemma}
\begin{proof}
	Please refer to Appendix A. 
\end{proof}

As observed from (\ref{beta}), we have $$ \beta_{1}\beta_{2} = \frac{Q U\Delta}{2} \approx \frac{D}{2d_0} \Delta.$$ 

Hence, the function $G(\cdot)$ can be rewritten as a function of $\{\Delta, \beta_2\}$. It is observed that if the array aperture is fixed, $ \beta_2 $ is only determined by the user location $ (r_{0}, \theta_{0})$, which indicates that each user location corresponds to a specific $ \beta_2 $.
Furthermore, as the user moves farther from the BS (consequently, the user is more likely to be located in the far-field region), the value of $ \beta_2 $ decreases.
Supposing that $ N = 257 $, $ U = 16 $ and $ f = 30 $ GHz and the user is located at the Fresnel and focusing region ($ \theta_{0} = 0 $), we have $ \beta_2 \in [1.68, 3.57] $.
In Fig.~\ref{fig:energySpread}, we numerically show the  function of $G(\cdot)$. Importantly, for any near-field users, the so-called energy-spread phenomenon \cite{nf_exhaustive} is still observed. Moreover, due to the period of $ \theta $, the energy-spread effects exist in each period interval ($ \frac{2}{U} $). Then, two key observations are obtained.
\begin{figure}
	\centering
	\includegraphics[width=\columnwidth]{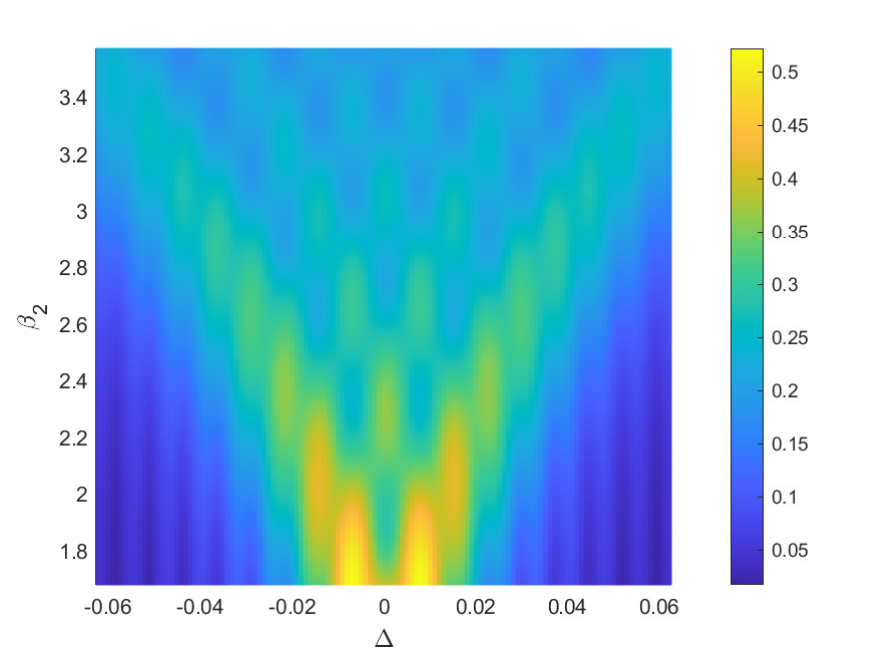}
	\caption{Received beam pattern in one period, where $N=257$, $ U = 16 $ and $f=30$ GHz.}
	\label{fig:energySpread}
\end{figure}
\begin{figure*}[t]
	\includegraphics[width=18.5cm]{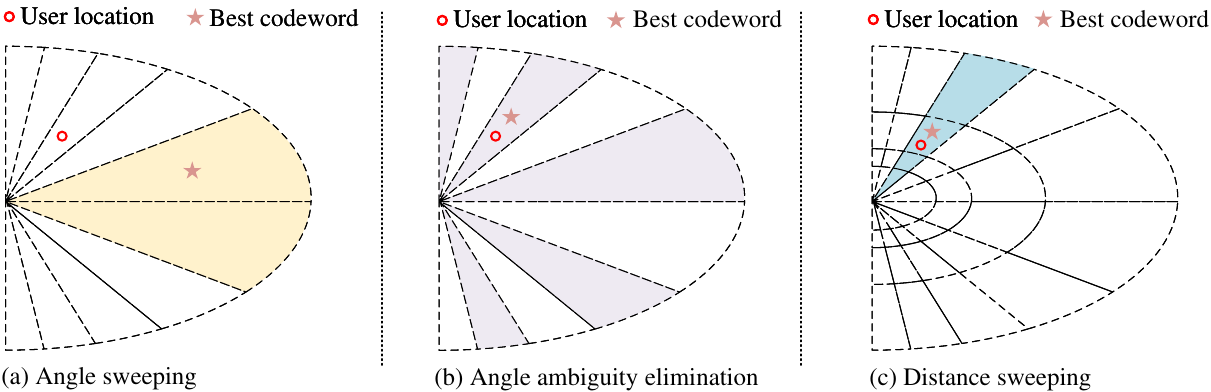}
	\centering
	\caption{Illustration of the proposed three-phase training scheme based on the sparse DFT codebook}
	\label{fig:trainingMethod}
\end{figure*}
\begin{observation}\label{OB1}\emph{In Figs.~\ref{fig:mixedFieldPeriod} and \ref{fig:energySpread}, the energy-spread effects exhibit in the whole radiating near-field region, and the received beam pattern with sampled beamforming vector $ \mathbf{w} $  contains useful user angle information:
		\begin{itemize}
			\item[1)] \textbf{User angle information:} It is observed that the actual user angle approximately locates in the \emph{middle} of the $3$ dB angular support in one period where $ \theta \in {\mathcal{L}_{0}} = [\theta_{0}-1/U, \theta_{0}+1/U] $. Mathematically, we have 
			\begin{equation}
				\theta_{0}\approx {\rm{Med}}(\mathcal{A}_{3}^{\mathcal{L}_0}(\theta_{0}, r_{0})).
			\end{equation}
			Moreover, we define $ u $-th angular support as $ \mathcal{A}_{3}^{\mathcal{L}_u}(\theta_{0}, r_{0}),$ where $ {\mathcal{L}_u} =  {\mathcal{L}_{0} + \frac{2u}{U}} $, $\forall u \in \mathcal{U} \triangleq \{\pm 1, \pm 2,\cdots, \pm (U-1)\}$.
			If we can employ the middle of $ u $-th angular support denoted by $\theta_{u}\approx {\rm{Med}}(\mathcal{A}_{3}^{\mathcal{L}_u}(\theta_{0}, r_{0}))$, the BS can infer that the user angle is among the following candidate angles $\theta_{u}\approx {\rm{Med}}(\mathcal{A}_{3}^{\mathcal{L}_u}(\theta_{0}, r_{0})),$ where $ {\mathcal{L}_{\rm c}} =  {\mathcal{L}_{u} + \frac{2c}{U}} $, $\forall c \in \mathcal{U}$,
			due to the period of $ \frac{2}{U} $.
			\item[2)] Although energy spread effects can provide actual user angle information, it is worth noting that when the user is located near the boundary of Fresnel region (in Fig. \ref{fig:energySpread}, the boundary corresponds to $ \beta_2 = 3.57 $ and $ r_0 = 7.42$m ), the angular support is distorted, which affects the accuracy of the beam training results.
			Moreover, by taking the noise and power fluctuation into account, the actual user angle may slightly deviate from the middle of the angular support.
		\end{itemize}
}
\end{observation}

\textbf{Observation}~\ref{OB1} indicates that the actual user angle can be estimated by finding the middle of the angular support in a period.
In other words, the sparse DFT codebook within the angle range of a period (for example $ \theta \in [-1/U,1/U)$) contained all information for which the BS can infer $ U $ candidate angles for actual user angle.

\section{Proposed Beam Training Scheme Enabled by Sparse DFT Codebook}
\label{Sec:scheme2}
In this section, we propose a three-phase near-field beam training method enabled by the sparse DFT codebook. Then, an optimization problem is formulated to minimize the beam training overhead.
\subsection{Phase 1: Beam Sweeping with the Sparse DFT codebook}
In Section~\ref{Sec:DFTbeamPattern}, we show that the received beam pattern with the sparse DFT codebook exhibits a period of $ 2/U $ and the user angle information can be inferred from the middle of the angular support within the region $ \mathcal{L}_0 $.

This motivates us to perform angular sweeping in a period of $ \theta $ (steering beams varying from $ -1/U $ to $ 1/U $) to estimate the potential user angles in the first phase, thereby further decreasing the beam training overhead. Considering that the beam width of the sparse far-field beamforming vector \eqref{eq:SLA steering vector} is $ \frac{4}{QU} $ \cite{wang2023can}, we sample the angular space as 
$$ \theta_s = \frac{2s - QU - 1}{QU}, s = 1,\cdots,QU. $$
Then, the sparse DFT codebook in the first phase to perform the angular sweeping is given by
\begin{center}
	\begin{framed}
		{\setlength\abovedisplayskip{2pt}
			\setlength\belowdisplayskip{2pt}
			\begin{equation} 
			\label{Eq:sampled DFT Codebook}
			{{\mathcal{W}}}_{\rm{DFT}}\!\! =\!\! \{{{\mathbf{w}}}_{\frac{QU-Q+1}{2}},\!\cdots\!,{{\mathbf{w}}}_g,\cdots,{{\mathbf{w}}}_{\frac{QU+Q-1}{2}} \},
			\end{equation} 
	}\end{framed}
\end{center}
where $ {{\mathbf{w}}}_g = \mathbf{a}_{\rm SLA}(\theta_g) $, $ \theta_g =  \frac{2g - QU - 1}{QU}$ and the index $ g \in \mathcal{G} \triangleq \{ \frac{QU-Q+1}{2} ,\cdots, \frac{QU+Q-1}{2}  \} $.

%Specifically, the propose beam training scheme based on Observation \ref{OB1} and Proposition \ref{pro:DFT periodicity} consists of the following three procedures including beam sweeping in the angular domain to estimate a potential user, angle ambiguity elimination to obtain actual user angle and beam sweeping in the user direction to obtain beam focusing gain. 
Then, the BS sequentially transmits $ Q $ pilot symbols with the sparse DFT codebook in \eqref{Eq:sampled DFT Codebook}, while it tunes beam angles varying from $ -1/U $ to $ 1/U $ as illustrated in Fig. \ref{fig:trainingMethod}(a). For each codeword, the received signal power at the user is given by 
\begin{equation}
	p(\mathbf{w}_g) = | \sqrt{Q}\beta \mathbf{b}^{H}(r_{0},\theta_{0})\mathbf{w}_g x + z|^2, ~\forall g\in\mathcal{G}.
\end{equation}

However, it is worth noting that when we transmit sparse far-field beamforming vectors varying from $ -1/U $ to $ 1/U $, the received beam pattern may exhibit a \emph{shifted} angular support within two periods instead of a whole angular support in one period as illustrated in Fig. \ref{fig:DFTbeamTraining}.
To obtain a regular angular support, the user needs to perform the received-beam-pattern shifting based on the index of the codeword $ \mathbf{w}_g $ with the lowest power.
We denote the codeword with the lowest power as $ \mathbf{w}_{\ell} $ corresponding to the angle $ \theta_{\ell} $.  
Specifically, we shift the angle larger than $ \theta_{\ell} $ by one period ($ 2/U $) as illustrated in Fig. \ref{fig:DFTbeamTraining}.
Then, the shifted indices of the codewords are arranged in a vector
$\mathbf{s} = [ {\ell-Q+1},\cdots, {\ell-Q+2},\cdots, \ell ]$.
Moreover, the equivalent received power of the shifted codewords can be recast as
\begin{equation*}
	\begin{aligned}
		\mathcal{P} = \{&p(\mathbf{w}_{\ell+1}),p(\mathbf{w}_{\ell+2}),\cdots,p(\mathbf{w}_{(QU+Q-1)/2}),\\
		&p(\mathbf{w}_{(QU-Q+1)/2}),p(\mathbf{w}_{(QU-Q+3)/2}),\cdots, p(\mathbf{w}_{\ell}) \}.
	\end{aligned}
\end{equation*} 
As such, we obtain a complete and regular angular support, which involves the user's angle information based on \textbf{Observation} \ref{OB1}. Specifically, the indices of shifted codewords with significantly high power received by the users are given by
\begin{equation}\label{Eq:phi}
	\mathcal{S} = \{\mathbf{s}_{\delta} |  p({\mathbf{w}_{\delta}})  > \kappa  \max \mathcal{P} \},
\end{equation}
where $\kappa \approx 0.5$. Then, the estimated angle is given by $ \theta_{\check{s}} = \frac{2\check{s} - QU - 1}{QU} $, where $ \check{s} = \rm{Med}(\mathcal{S}) $ and the corresponding codeword is $ \mathbf{w}_{\check{s}} $. 
%However, it is noteworthy that when we transmit sparse far-field beamforming vectors varying from $ -1/U $ to $ 1/U $, the received beam pattern may exhibit a \emph{shifted} angular support within two periods instead of a whole angular support  in one period as illustrated in Fig \ref{fig:DFTbeamTraining}.
%As such, it is necessary to shift the received beam pattern by a period ($ \frac{2}{U} $) along the codeword with the lowest received power/SNR to obtain a regular angular support.
Then, according to the periodicity, the BS can infer $ U $ candidate user angles, which are given by
\begin{equation}
\label{eq:candidate user angle}
	\theta_{\rm c} = {\theta}_{\check{s}}+\frac{2u}{U}, \forall u \in \mathcal{U}.
\end{equation}
%Next, we introduce the angular ambiguity elimination method based on a central sub-array in the second phase.
%Therefore,we can just do the beam training in a small area with a angular range from $ -\frac{1}{U} $ to $ \frac{1}{U} $ by the sparse DFT codebook to find a potential angle and then eliminate the angle ambiguity via a sub-array.
\subsection{Phase 2: Angular Ambiguity Elimination}
\label{sec:phase2}
In the second phase, we propose an efficient method to resolve the angular ambiguity. The key idea is to utilize a central subarray to sequentially examine the candidate angles in \eqref{eq:candidate user angle}.
%As shown in Section \ref{Sec:DFTbeamPattern}, there exists angular ambiguity of the user due to the periodicity of the sparse DFT codebook.
Specifically, we activate a central subarray with $ M $ antennas to eliminate the angular ambiguity, for which the codebook $ {{\mathcal{V}}}_{\rm{Sub}} $ is presented as follows
\begin{center}
	\begin{framed}
		{\setlength\abovedisplayskip{2pt}
			\setlength\belowdisplayskip{2pt}
			\begin{equation} 
			\label{Eq:sub-array DFT Codebook}
			{{\mathcal{V}}}_{\rm{Sub}}\!\! =\!\! \{{{\mathbf{v}}}_{1},\!\cdots\!,{{\mathbf{v}}}_s,\cdots,{{\mathbf{v}}}_{QU} \},
			\end{equation} 
	}\end{framed}
\end{center}
where $ {\mathbf{v}}_{s}$ is given by
\begin{equation}
{\mathbf{v}}_{s} =\!\!\frac{1}{\sqrt{M}}\bigg[\underbrace{0,\cdots, 0}_{\frac{N-M}{2}},e^{\jmath {\pi}\frac{N-M}{2}\theta_s},\!\cdots\!,e^{\jmath{\pi}\frac{N+M-2}{2}\theta_s},\underbrace{0,\cdots, 0}_{\frac{N-M}{2}},\bigg]^H\!\!.
\label{eq4}
\end{equation}
\begin{figure}
	\centering
	\includegraphics[width=\columnwidth]{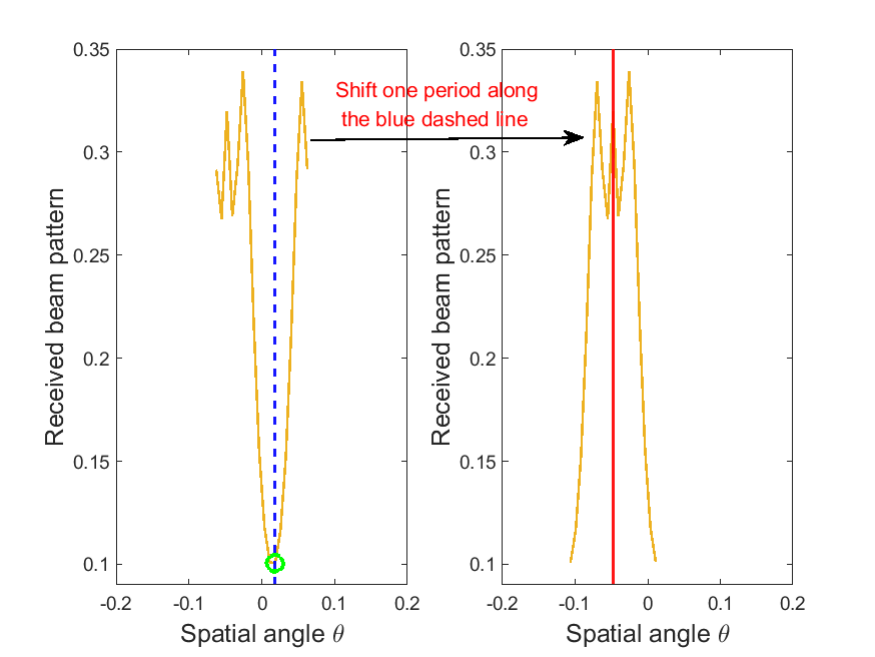}
	\caption{The shift of the received beam pattern with the sparse DFT codebook in a period, where $N=257$, $ U = 16 $ and $f=30$ GHz. The actual user angle is $ \theta_0 = 0.2 $}
	\label{fig:DFTbeamTraining}
\end{figure}

It is noteworthy that there are two criteria for ensuring the effectiveness of angular ambiguity elimination.
The first one is that the user should be located in the far-field region of the activated central subarray to avoid the energy-spread effect.
Moreover, considering that the interval between two adjacent candidate angles is $ \frac{2}{U} $, extra angular ambiguity will be introduced if the beam width of the central subarray is wider than $ \frac{4}{U} $.
Therefore, the second criterion is that the beam width of the central subarray is required to be smaller than $ \frac{4}{U} $, which ensures no interference between two adjacent candidate angles.

To satisfy the first criterion, we set the Rayleigh distance of the subarray to be less than the Fresnel distance of the XL-array, which is given by
\begin{equation}
\label{eq:far-field condition}
\frac{2M^2d_0^2}{\lambda} \le 1.2D.
\end{equation}
Then, \eqref{eq:far-field condition} can be simplified as $ M \le \sqrt{1.2{(N-1)}} $.

For the second criterion, we have 
\begin{equation}
	 \frac{4}{M} \le \frac{4}{U},
\end{equation}
which can be simplified to $ M \ge U $.
Therefore, the number of antennas of the central subarray needs to satisfy
\begin{equation}
\label{antenna numer condition}
U \le	M \le \sqrt{{1.2(N-1)}}.
\end{equation}

Moreover, from \eqref{antenna numer condition}, the activation interval $ U $ has the following constraint $ U \le \sqrt{1.2{(N-1)}} $.
For example, given a setup where $ N = 257 $ and $ U = 16 $, we have $16 \le M \le 17.5 $. In Fig. \ref{fig:subArrayIllustration}, we plot the beam width of the subarray with $ M = 17 $.
It is observed that two adjacent candidate user angles have no considerable mutual interference, which indicates the effectiveness of angular ambiguity elimination.

\begin{figure}
	\centering
	\includegraphics[width=\columnwidth]{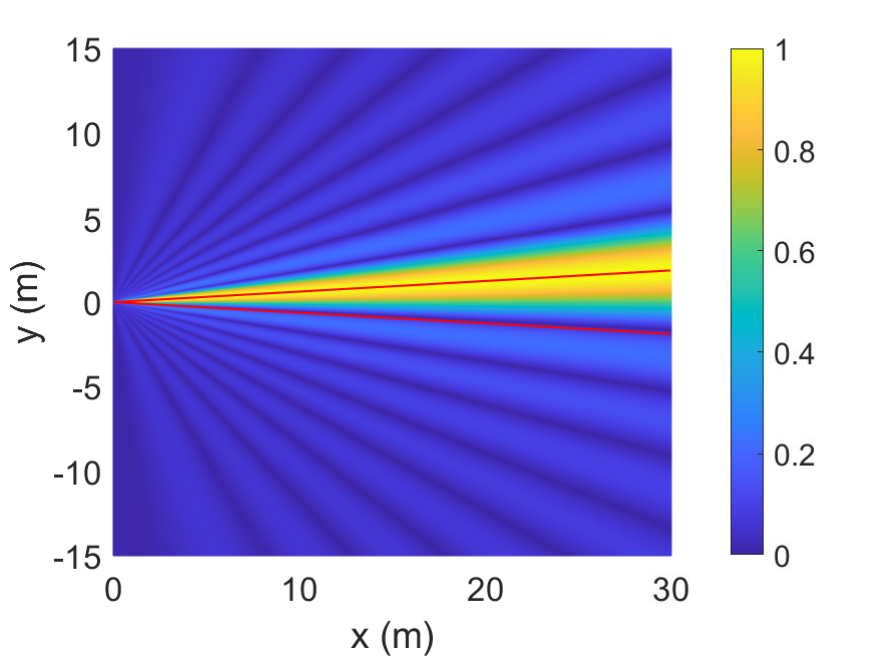}
	\caption{Illustration of angle ambiguity eliminated by the central subarray, where $N=257$, $ M = 17 $ and $f=30$ GHz. The beam is steered towards an candidate angle $ 1/16 $. The red lines are two adjacent candidate user angles $ -1/16 $ and $ 1/16 $.}
	\label{fig:subArrayIllustration}
\end{figure}
Based on the above, the BS activates a central subarray comprising $ M $ antennas satisfying \eqref{antenna numer condition} to resolve the angular ambiguity.
For the candidate user angles $ {\theta}_{\check{s}} + u\frac{2}{U}, \forall u \in \mathcal{U}$ in \eqref{eq:candidate user angle}, the codewords chosen from \eqref{Eq:sub-array DFT Codebook} to resolve the angular ambiguity are $ \mathbf{v}_{\check{s}+uQ}, \forall u\in \mathcal{U} $.
Then, the BS sequentially transmits $ U $ pilot symbols with beamforming vectors $ \{ \mathbf{v}_{\check{s}+uQ} \} $ as illustrated in Fig. \ref{fig:trainingMethod} (b).
For each codeword, the received signal power at the user in the second phase is given by 
\begin{equation}
	\label{eq:power teo-phase}
	p(\mathbf{v}_{\check{s}+uQ}) = | \sqrt{M}\beta \mathbf{b}^{H}(r_{0},\theta_{0})\mathbf{v}_{\check{s}+uQ} x + z|^2, ~\forall u\in\mathcal{U}.
\end{equation}

Then, we can obtain the optimal codeword $ \mathbf{v}_{s^{\ast}} $ via simple comparisons of received signal power in \eqref{eq:power teo-phase}, which is given by
\begin{equation}
	{s}^\ast =\arg \max _{u \in \mathcal{U}} p(\mathbf{v}_{\check{s}+uQ}).
\end{equation}
Then, the estimated user angle can be obtained by $ \theta_{s^{\ast}}= \frac{2s^{\ast} - QU - 1}{QU}$.
\subsection{Phase 3: Beam Sweeping with Polar-domain Codebook}
\label{sec:phase 3}
Once we determine the user angle $ \theta_{s^{\ast}} $ in the second phase, we can use the polar-domain codebook to perform the range domain sweeping, achieving beamfocusing gain in the near-field.
The polar-domain codebook utilized in the third phase is given by
\begin{equation}
	\label{eq:polar-domain codebook}
	\bar{\mathcal{X}}_{\rm{Pol}} = \{{\mathcal{X}}_1,\cdots,{\mathcal{X}}_s,\cdots,{\mathcal{X}}_{QU}\},
\end{equation} 
where $ {\mathcal{X}}_s = \{{\mathbf{x}}_{s,1},\cdots,{\mathbf{x}}_{s,v},\cdots, {\mathbf{x}}_{s,V}\} $ and $ {\mathbf{x}}_{s,v} = \mathbf{b}^{H}(r_{s,v},\theta_{s})$ with $ r_{s,v}=\frac{1}{v} \alpha_{\Delta}\left(1-\theta_{s}^2\right)$  \cite{nf_exhaustive}.
Moreover, $ V $ denotes the number of range samples.

Specifically, the codewords employed in the third phase are given by $ \mathcal{X}_{s^{\ast}} = \{{\mathbf{x}}_{s^{\ast},1},\cdots,{\mathbf{x}}_{s^{\ast},v},\cdots, {\mathbf{x}}_{s^{\ast},V}\} $. 
Then, the BS activates the whole XL-array and sequentially transmits training $ V $ symbols in the estimated user angle $ \theta_{s^{\ast}} $ with codewords $ \mathcal{X}_{s^{\ast}} $ as illustrated in Fig \ref{fig:trainingMethod}(c).
For each codeword, the received signal power at the user is given by 
\begin{equation}
	\label{eq:power third-phase}
	p({\mathbf{x}}_{s^{\ast},v}) = | \sqrt{N}\beta \mathbf{b}^{H}(r_{0},\theta_{0}){\mathbf{x}}_{s^{\ast},v} x + z|^2, ~\forall v\in\mathcal{V}.
\end{equation}
Subsequently, the optimal codeword is determined through straightforward comparisons of received signal power in \eqref{eq:power third-phase} and the best codeword index is given by
\begin{equation}
	{v}^\ast =\arg \max _{v \in \mathcal{V}} p({\mathbf{x}}_{s^{\ast},v}).
\end{equation}
Therefore, we can obtain the user location $ (r_{{s}^\ast,{v}^\ast}, \theta_{{s}^\ast}) $ with $ r_{{s}^\ast,{v}^\ast} = \frac{1}{{v}^\ast} \alpha_{\Delta}\left(1-\theta_{{s}^\ast}^2\right) $.

The detailed procedures of the proposed three-phase beam training method are summarized in Algorithm \ref{Algorithm:three-phase}.
\begin{algorithm}
	\caption{Proposed Three-phase Beam Training Method}
	\label{Algorithm:three-phase}
	\begin{algorithmic}[1]  
		\STATE{\textbf{Phase 1: Angular sweeping in the subspace.}} 
		\STATE {Use the sparse DFT codebook $ \mathcal{W}_{\rm DFT} $ for the beam sweeping in the angular subspace $ [-1/U, 1/U) $.}
		\STATE {Perform the received-beam-pattern shifting to obtain a regular angular support $ \mathcal{P}$.}
		\STATE {Obtain $ U $ candidate user angles $ \theta_{\rm c} = {\theta}_{\check{s}}+\frac{2u}{U}, \forall u \in \mathcal{U} $ according to the middle of the angular support $ \mathcal{P}$.}
		\STATE{\textbf{Phase 2: Resolve angular ambiguity.}} 
		\STATE {Activate a central subarray with $ M $ antennas and use codebook $ {{\mathcal{V}}}_{\rm{Sub}} $ to examine $ U $ candidate angles estimated in the first phase.}
		\STATE {Obtain the best user angle $ \theta_{s^\ast} $ according to the highest power received at the user with respective to $ p(\mathbf{v}_{\check{s}+uQ}) $.}
		\STATE{\textbf{Phase 3: Range sweeping.}} 
		\STATE{Use polar-domain codebook $ \bar{\mathcal{X}}_{\rm{Pol}} $ with codewords $ \mathcal{X}_{s^{\ast}} $ to sweep the range domain in the user angle $ \theta_{s^\ast} $}.
		\STATE {Obtain the best user range $ r_{s^\ast,v^\ast} $ according to the highest power received at the user with respective to $ p({\mathbf{x}}_{s^{\ast},v}) $.}
	\end{algorithmic}
\end{algorithm}
\subsection{Discussions }
\label{sec:discussion}
{\underline{\textbf{Beam training overhead:}}}
In the first phase, the overhead is $ Q = \frac{N-1}{U}+1 $.
Moreover, the overhead of the second and third phases are $ U $ and $ V $, respectively.
Finally, the overall overhead of the proposed beam training scheme is $T^{\rm (3P)} =  Q+U+V $.
Given $ Q = \frac{N-1}{U}+1 \in \mathbb{Z}^{+} $ with $ \mathbb{Z}^{+} $ denoting the positive integer set, the beam training overhead of the proposed method enabled by the sparse DFT codebook can be recast as 
\begin{equation}
	\label{Eq:multi-beam training overhead}
	T^{\rm (3P)} = \frac{N-1}{U} + U + V + 1.
\end{equation}
It can be easily obtained from \eqref{Eq:multi-beam training overhead} that as $ U $ increases, the overhead during the first phase decreases, while it increases during the subsequent phase, and vice versa. Thus, there is a fundamental trade-off between the first and second phase of the beam training method depending on the value of $ U $.
Then, we aim to minimize the beam training overhead $ T^{\rm (3P)}  $ and the optimization problem is formulated as
\begin{subequations}
		\begin{align}
			({\bf P1}):~~~& \min_U~ F(U) = \frac{N-1}{U} + U +V+1 \nn\\
			&~~{\text{s}}{\text{.t}}{\rm{. }} ~ U \le \sqrt{{1.2(N-1)}},\\
			&~~~~~~U \in \mathcal{F} \triangleq \{\frac{N-1}{U}\!+\!1\in \mathbb{Z}^{+}, U \in \mathbb{Z}^{+}\} \label{Eq:interger constraint}.
		\end{align}
\end{subequations}If we remove the integer constraint \eqref{Eq:interger constraint}, $ ({\bf P1}) $ is a convex problem. Then, Problem $ ({\bf P1}) $ can be rewritten as 
\begin{subequations}
	\begin{align}
	({\bf P2}):~~~& \min_U~ F(U) = \frac{N-1}{U} + U +V+1 \nn\\
	&~~{\text{s}}{\text{.t}}{\rm{. }} ~ U \le \sqrt[4]{1.2{(N-1)}}.
	\end{align}
\end{subequations}

Problem ({\bf P2}) is an convex optimization problem, whose optimal solution can be easily obtained as follows due to convexity.
\begin{lemma}
{\rm
	The optimal solution to Problem ({\bf P2}) is given by {\small$ U^\ast_{P_2} = \sqrt{(N-1)} $}. Moreover, the overhead of the proposed method is {\small$ T^{(\rm 3P)} = 2\sqrt{(N-1)}+V+1 $}.
}
\end{lemma}
\begin{proof}
	It is observed that {\small $ \frac{N-1}{U} + U +V+1 \ge 2\sqrt{N-1}+V+1 $, where the equal holds with $ U = \sqrt{(N-1)} $}.
	Moreover, considering $ U = \sqrt{(N-1)} \le  \sqrt{1.2{(N-1)}}$, the optimal solution to Problem $ ({\bf P2}) $ is  {\small$ U^\ast_{P_2} = \sqrt{(N-1)} $} with the beam training overhead {\small$ T^{(\rm 3P)} = 2\sqrt{(N-1)}+V+1 $}.
\end{proof}

When the integer constraint \eqref{Eq:interger constraint} is taken into account, a suboptimal solution to Problem ({\bf P1}) is given by
\begin{equation}
\hat U = \arg \min\limits_{\mathcal{F}}\{  |U^{\ast}_{\rm P_2} - f|, f \in  \mathcal{F} \}.
\end{equation}
%We denote {\small$ U_{\rm left} = \arg \min\limits_{\mathbb{F}} \{\sqrt{{N-1}}-f|f\in\mathbb{F}, f < \sqrt{N-1}\} $} and {\small$ U_{\rm right} = \arg \min\limits_{\mathbb{F}} \{f-\sqrt{{N-1}}|f\in\mathbb{F}, f > \sqrt{{N-1}}\} $}
%Then, the optimal $ U $ is obtained by {\small$ U^{\ast} = \arg \min\limits_{\{U_{\rm left},U_{\rm right}\}} \{F(U_{\rm left}),F(U_{\rm right})\}$}.

Next, we present an example to illustrate the low overhead feature of the proposed beam training method.
We consider a setup with $ N = 1025 $ and $ V = 5 $. Then, $ U =32 $ can minimize $ T^{(\rm 3P)} $ according to the optimal solution to Problem $ ({\bf P1}) $.
The beam training overhead of the proposed multi-beam training scheme in this setup is $ T^{(\rm 3P)} = 2\sqrt{N-1}+V+1 = 64+5+1=70 $, which is significantly smaller than that of the exhaustive-search method ($ T^{(\rm EX)} = 5280 $) and the two-phase near-field beam training method ($ T^{(\rm 2P)} = 1061 $).

\begin{remark}[Improved scheme: Middle-$K$ angle selection]\label{Re:SelectionK}
	\emph{
		Due to the power fluctuation \cite{two_phase}, the middle angle of the angular support may not be accurate. 
		To improve the estimation accuracy, we can select the \emph{middle-$ K $ angles} of the angular support instead of selecting one potential angle in the first procedure.
		Specifically, in the third phase, we should perform the beam sweeping in the range domain in the $ K $ potential angles $ \theta_{\bar{s}_k}, k = 1,\cdots,K $ to determine an optimal polar-domain codeword.
		The overhead of the middle-$ K $ angle scheme is given by $ T^{(\rm 3P)}_{{\rm Mid}-K} = \frac{N-1}{U} + U + KV$.
		It is noteworthy that this does not significantly increase the beam training overhead, which is still proportional to $ \sqrt{N} $.
	}
\end{remark}

\begin{remark}[Estimation error] \label{remark:estimation_error}
	\emph{From the above analysis, the accuracy of the proposed beam training scheme enabled by sparse DFT codebook is mainly dependent on two factors.
	First, the sampling resolution is a key factor because the proposed beam training scheme is an on-grid channel estimation method.
	Therefore, when the sampling interval is small enough, the performance of proposed multi-beam training scheme approaches that of the optimal beamformer without noise taking into account.
	Second, noise is another key issue.
	The estimated angle is derived from the received signal power in different time, consequently influenced by the received SNR. 
	Hence, a higher SNR is expected to achieve more accurate angle estimation,
	which will be numerically verified in Section \ref{Sec:numericalResults}. 
}
\end{remark}

\begin{remark}[Multi-path channels]
	{\rm
		With respect to multi-path channel cases, we divide the extension of the proposed algorithm into two cases according to different Rician parameters.
		\begin{itemize}
			\item \textbf{LoS-dominant channel:} When the Rician factor is large enough (or equivalently the LoS path is dominant), we can regard the NLoS components as environmental noise.
			Therefore, the proposed beam training scheme based on the sparse DFT codebook still holds as this method only depends on the LoS path component.
			\item \textbf{Comparable multi-path components:} This case is much more complicated. Considering that the accuracy depends on the periodic angular supports, NLoS components may bring about randomly overlapped received beam pattern in the angular domain, which poses challenges to the angle estimation in the first phase of our proposed beam training scheme.
			Therefore, the case with comparable multi-path components is left as a topic for our future works.
		\end{itemize}
	}
\end{remark}

\begin{remark}[Universal in both near- and far-field communications] \emph{The proposed beam scheme can be applied to both near- and far-field communications.
We can identify the near- or far-field user according to the angular support width \cite{wuxun}.
Specifically, in cases where the angular support width is small (e.g., only contains one candidate angle), it signifies that the user is located in the far-field region.
Then, there is no need to perform the range estimation in the third phase (see Section \ref{sec:phase 3}).
In other words, for far-field users, only the first two phases of the proposed beam training scheme need to be executed.
}  
\end{remark}
\begin{figure}
	\centering
	\includegraphics[width=\columnwidth]{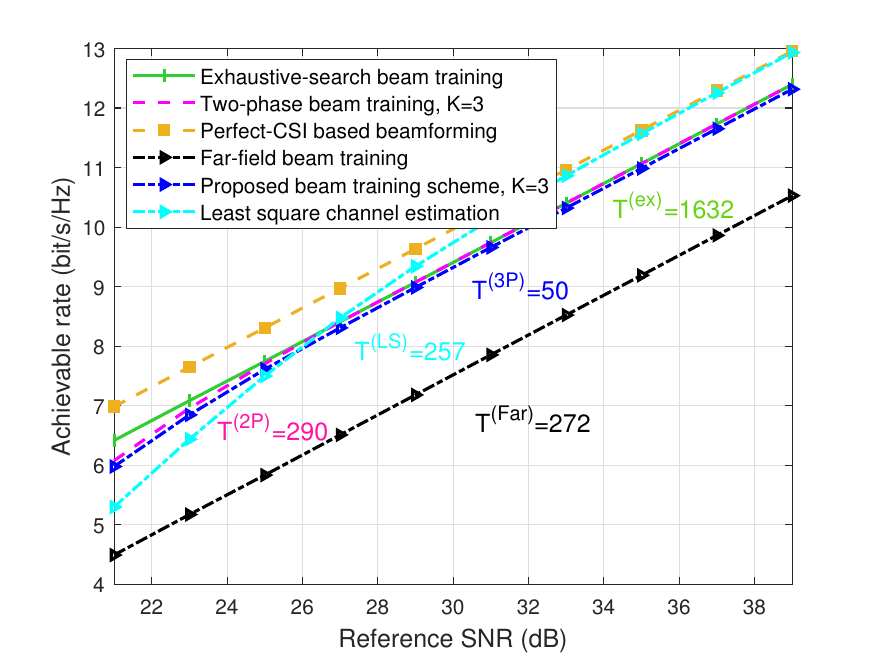}
	\caption{Achievable rate versus SNR.}
	\label{fig:Achievable rate versus SNR}
\end{figure}
\section{Numerical Results}
\label{Sec:numericalResults}
Numerical results are provided to validate the effectiveness of the proposed near-field beam training scheme in this section.
We first present the system parameters and benchmark schemes followed by the performance comparison under numerous setups.
\subsection{System Setup and Benchmark Schemes}
We set the system parameters as follows.
We consider that the XL-array is equipped with $ N = 257 $ antennas and $ f = 30 $ GHz.
The transmit power and reference channel gain at $ 1 $ m are set as $ P_{\rm tol} = 30 $ dBm and $ \beta_0 = (\frac{4\pi}{\lambda})^2  = -62 $ dB, respectively.
Moreover, the noise power is $ \sigma^2 = -80 $ dBm.
According to $ ({\bf P1})  $, the optimal activation interval is set as $ U = 16 $.
Furthermore, the antenna number of the activated central subarray is $ M = 17 $.
With respective to NLoS paths, we set $ L = 2$ and $ \kappa_k = 30$ dB \cite{dl,chirp_hier}.
The reference SNR for a user is defined by $ \gamma = \frac{NP_{\rm tol} \beta_0}{r_{\rm 0}^2 \sigma^2} $ \cite{two_phase}.
To characterize the overhead, we assume that the total transmission time and a pilot symbol time are $ T_{\mathrm{tol}} = 0.2 $ ms and $ T_{\mathrm{s}} = 0.1 $ \si{\micro\second} \cite{liu2024near}, respectively.
Then, the \emph{effective} rate is defined by $ R_{\mathrm{Eff}}=\left(1-\frac{T_{\rm over}T_{\mathrm{s}}}{T_{\mathrm{tol}}}\right) R $, where $ T_{\rm over} $ denotes the beam training overhead of each scheme.
All the numerical results are averaged over 1000 channel realizations.
The following benchmark schemes are considered for performance comparison:
\begin{figure}
	\centering
	\includegraphics[width=\columnwidth]{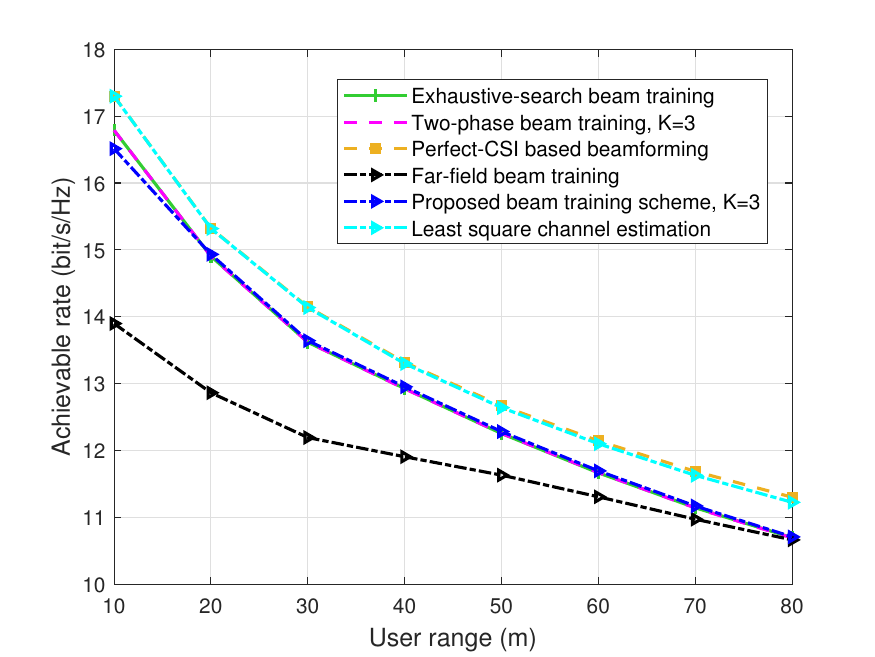}
	\caption{Achievable rate versus user range.}
	\label{fig:Achievable rate versus user range}
\end{figure}
\begin{itemize}
	\item {\rm \textbf{Perfect-CSI based beamforming:}}
	This scheme assumes that the BS perfectly aligns the user near-field channel and beamforming vector is set by $ \bar{\mathbf{w}} =  \mathbf{b}(r_{0},\theta_{0})$
	Obviously, this scheme is the performance upper bound for all methods.
	\item {\rm \textbf{Least square channel estimation:}}
	%We consider the downlink least square (LS) channel estimation with $ N_{\rm RF} = 1 $ as the benchmark scheme.
	%Specifically, the BS adopts the $ N \times N $ DFT matrix as the analog beamforming matrix over $ N $ pilot symbols, based on which the user estimates the BS-user channel in the downlink via LS estimation. The overhead of this scheme is $ T^{({\rm LS})} = N $. 
	This scheme is a classic off-grid channel estimation method where the user estimates the channels by $ N $ pilot symbols transmitted by the BS in the downlink.
	The estimated channel is given by $ \hat{\mathbf{h}}_{\mathrm{LS}}=\left(\mathbf{X}^{\mathrm{H}} \mathbf{X}\right)^{-1} \mathbf{X}^{\mathrm{H}} \mathbf{y} $, where $ \mathbf{X} \in \mathbb{C} ^{N \times N}$ denotes the pilot matrix for each user and $ \mathbf{y} = \mathbf{X}\mathbf{h}+\mathbf{z} \in\mathbb{C}^{N \times 1} $ represents the received signal vector with $ \mathbf{z} \sim \mathcal{C N}\left(\mathbf{0}, \sigma^2 \mathbf{I}\right) $. 
	Moreover, $ \mathbf{h} $ is shown in \eqref{Eq:nf-channel}.
	Obviously, the pilot overhead of this scheme is $ T^{({\rm LS})} = N $.
	\item {\rm \textbf{Exhaustive-search beam training:}}
	This scheme is detailed in Section \ref{Sec:Exhaustive Search}.
	Due to the different angle sampling interval, the overhead of this scheme is modified by $T^{\rm{(ex)}}=QUV$.
	\item {\rm \textbf{Two-phase beam training:}}
	This scheme is detailed in Section \ref{Sec:Two phase}.
	Due to the different angle sampling interval, the overhead of this scheme is revised by $T^{\rm{(2P)}} = QU+KV$.
	\item {\rm \textbf{Far-field beam training based on DFT codebook:}}
	Conventional DFT codebook is used to sweep the whole angular domain for choosing a best codeword for which the maximum received signal power is achieved at the user.
	%This scheme is the upper bound of the multi-beam training scheme based on sub-array.
	The beam training overhead of this scheme is $ T^{\rm (Far)} = QU $.
\end{itemize}
\begin{figure}
	\centering
	\includegraphics[width=\columnwidth]{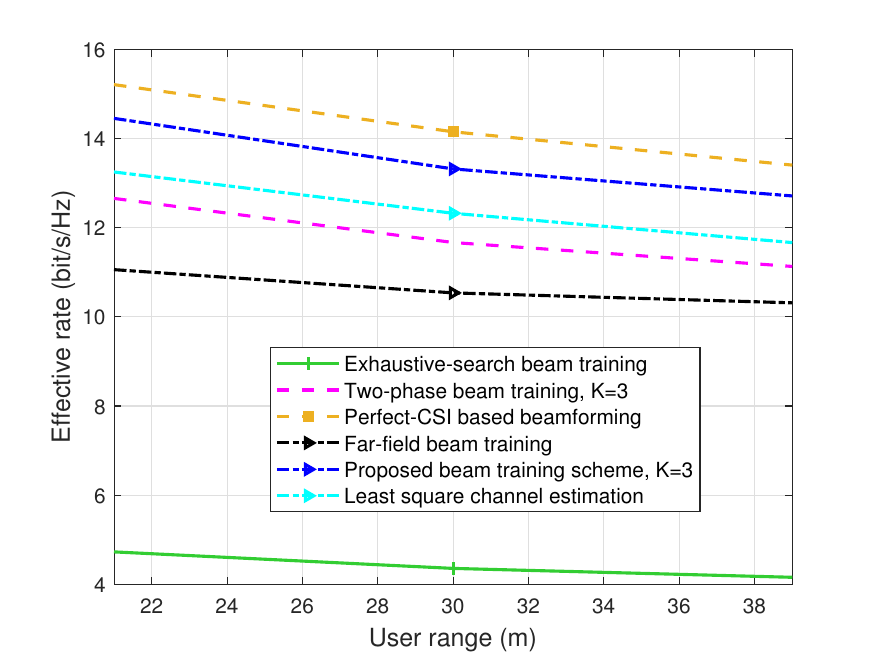}
	\caption{Effective rate versus user range.}
	\label{fig:Effective rate versus user range}
\end{figure}
\begin{figure}
	\centering
	\includegraphics[width=\columnwidth]{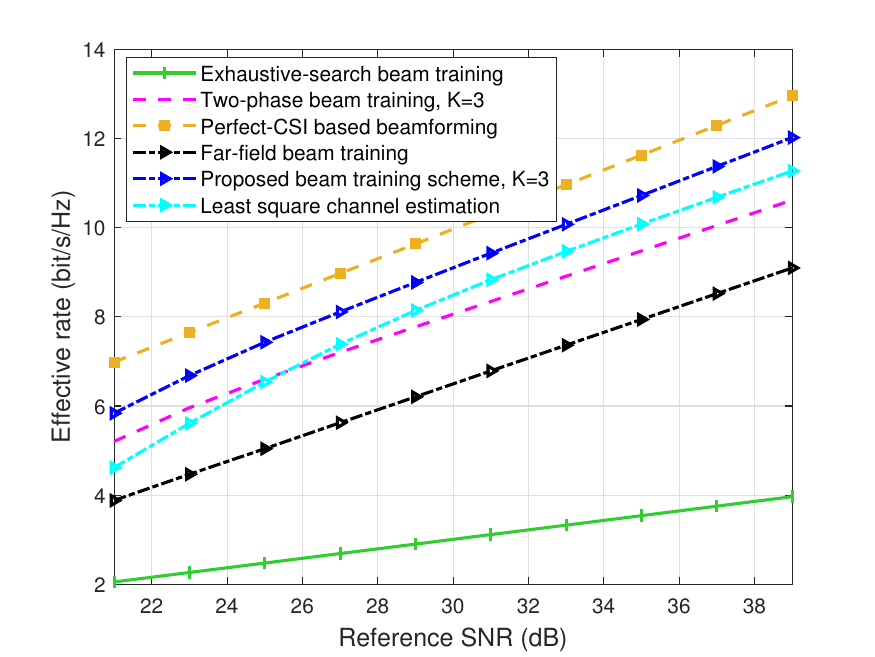}
	\caption{Effective rate versus SNR.}
	\label{fig:Effective rate versus SNR}
\end{figure}
\subsection{Performance Analysis}
In Fig. \ref{fig:Achievable rate versus SNR}, we plot the achievable rate $R$ versus the reference SNR $\gamma$ under different beam training schemes. Some key observations can be concluded as follows. First, our proposed near-field beam training scheme achieves very close performance to the two-phase beam training and exhaustive search based beam training scheme, especially in the high-SNR regime (i.e., larger than 26 dB). Second, with the decrease of the reference SNR, the achievable rate attained by off-grid channel estimation degrades more dramatically than other schemes, and becomes inferior to our proposed scheme when the reference SNR is lower than 26 dB. This is because the direct channel estimation is more sensitive to the received SNR as the XL-array beamforming vector may not be well aligned with the channel path during signaling. Finally, there is a large performance gap between the conventional far-field beam training scheme and other schemes dedicated to near-field communications, which implies that far-field beam training is no longer effective for next generation wireless systems with more antennas.

Then, Fig. \ref{fig:Effective rate versus SNR} illustrates the effective rate $ R_{\mathrm{Eff}}$ versus the reference SNR $\gamma$. Interestingly, our proposed scheme outperforms other schemes in terms of the effective rate except for the perfect-CSI based beamforming. This is because our scheme achieves close or even superior achievable rates compared with other benchmarks (as have shown in Fig. \ref{fig:Achievable rate versus SNR}), but with far less training overhead. Moreover, the exhaustive search based near-field beam training scheme is not practically applicable due to small effective rate caused by large training overhead, as shown in Fig. \ref{fig:Effective rate versus SNR}.
\begin{figure}
	\centering
	\includegraphics[width=\columnwidth]{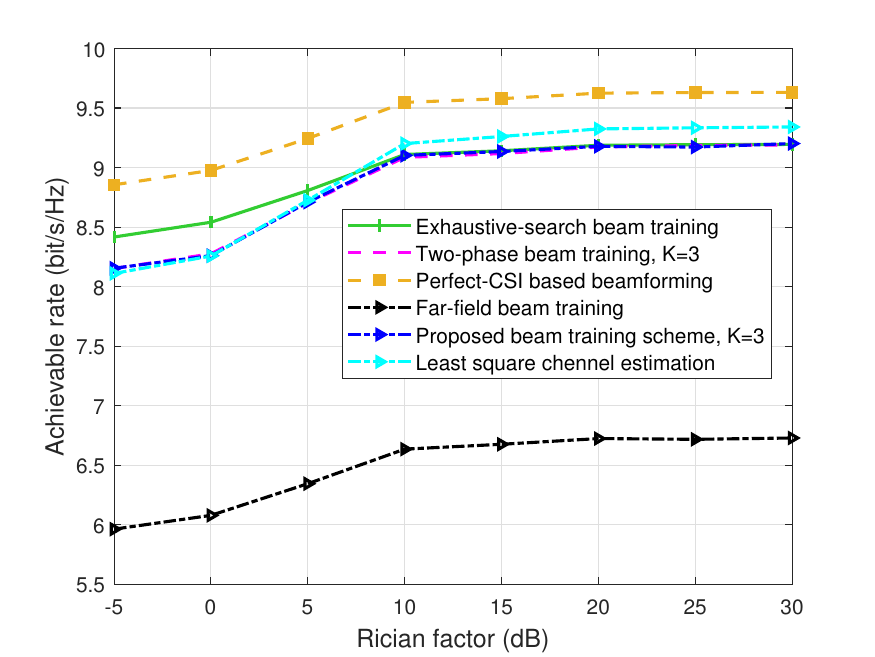}
	\caption{Achievable rate versus Rician factor.}
	\label{fig:Achievable rate versus Rician factor}
\end{figure}
\begin{figure}
	\centering
	\includegraphics[width=\columnwidth]{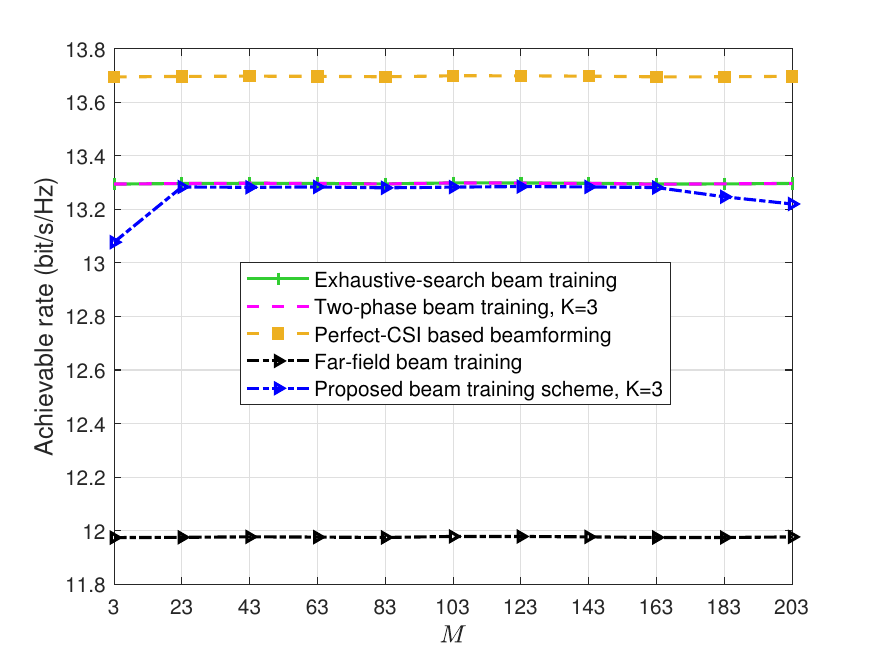}
	\caption{Achievable rate versus the number of antennas in the central subarray.}
	\label{fig:Achievable rate versus subarray}
\end{figure}

Fig. \ref{fig:Achievable rate versus user range} shows the effect of user range on the achievable rate. It can be observed that the proposed scheme with $K=3$ exhibits approximately the same performance as the exhaustive-search based and two-phase near-field beam training schemes for all user ranges. This is attributed to the smart design of the sparse DFT codebook. This method leads to a periodical energy-spread effect during beam training, and hence the former key observation in \cite{two_phase} can be leveraged for beam training. Second, the achievable rate performance of the proposed scheme largely outperforms the far-field beam training when the user range is less than 50 m and gradually converges to that of the far-field beam training. This verifies the universality of our proposed scheme for both near-field and far-field beam training cases. In addition, the relationship between effective rate and user range is depicted in Fig. \ref{fig:Effective rate versus user range}. Considering the overhead of beam training, the effective rate attained by our scheme is only slightly lower than (less than 1 bps/Hz) that of the perfect-CSI based schemes, which further verifies its effectiveness. 

In Fig. \ref{fig:Achievable rate versus Rician factor}, we evaluate the impact of Rician factor on the system achievable rate. It can be observed that the achievable rates of all schemes increase with the Rician factor
at first and gradually saturate when the Rician factor approximates 10 dB. Moreover, the two-phase beam training scheme slightly outperforms the proposed scheme when the Rician factor is less than 5 dB. This is because our scheme is more sensitive to noise due to operations such as received-beam-pattern shifting, since NLoS components can be treated as a form of environment noise.
\begin{figure}
	\centering
	\includegraphics[width=\columnwidth]{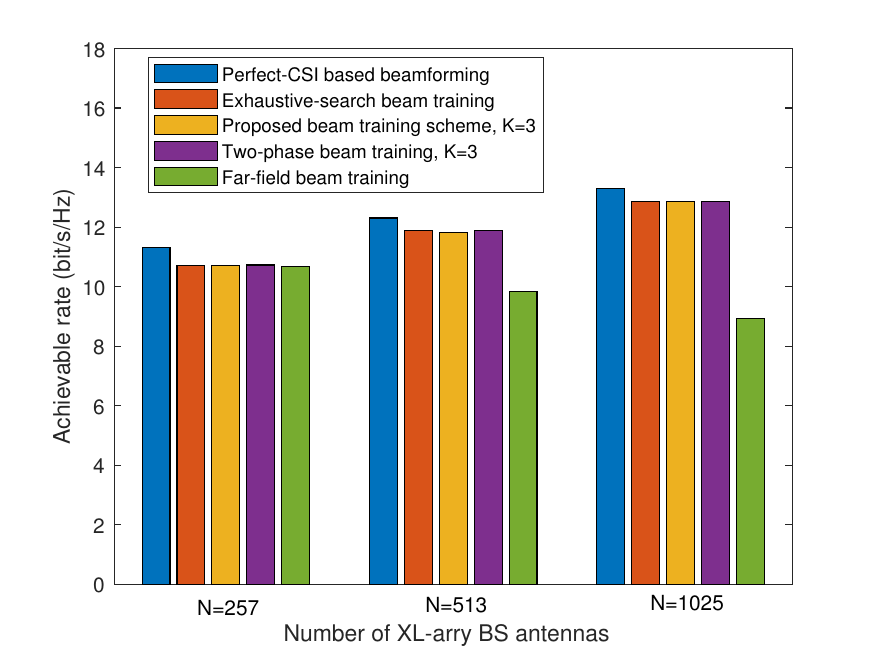}
	\caption{Achievable rate versus number of antennas.}
	\label{fig:Achievable rate versus atenna number}
\end{figure}

In Fig. \ref{fig:Achievable rate versus subarray}, we plot the achievable rate versus the number of antennas in the central subarray.
It can be observed that the achievable rate of the proposed scheme suffers from significant performance loss when the number of antennas in the central subarray is sufficiently small or large.
This can be explained by two facts: 1) When $ M $ is sufficiently large (Violation of Criteria 1 in Section \ref{sec:phase2}), the user are more likely to be located in the near-field region of the central subarray, for which the energy-spread effect is dominant and thus results in significant performance loss; 2) When $ M $ is sufficiently small (Violation of Criteria 2 in Section \ref{sec:phase2}), the beam width of the central subarray becomes too large, introducing extra ambiguity and thus failing to distinguish candidate user angles.

Finally, we plot the achievable rate versus the number of XL-array antennas in Fig. \ref{fig:Achievable rate versus atenna number}, where the user ranges are fixed at 80 m. It is observed that the achievable rates  increase as there are more antennas for all schemes except for the far-field beam training scheme. This is because as the number of antennas increases, the near-field effect is more prominent, and hence the far-field beam training is no longer effective even when the user is located at a relatively far range (i.e., 80 m) from the XL-array.
 
\section{Conclusion}
\label{Sec:Conclusion}
In this paper, we proposed a novel near-field beam training scheme enabled by the sparse DFT codebook (sparse far-field beamforming vectors) to construct periodic received beam pattern at the user.
To this end, we showed that the angular periodicity of the received beam pattern boosts reduction in sweeping space, thereby significantly decreasing beam training overhead.
Specifically, the middle of the angular support within a period contains the user angle information.
Then, an activated central subarray can resolve the angular ambiguity followed by polar-domain codebook sweeping in the best user angle.
Finally, numerical results were presented to show that the proposed beam training scheme can achieve nearly the same performance in the high-SNR regime with the exhaustive-search scheme, while significantly reducing the beam training overhead.

\section*{Appendix A}
\section*{Proof of Lemma \ref{lemma:LSAgeneral}}
From \eqref{eq:GeneralMixedField}, when $ \Delta \in  [-1/U,1/U)$, $ B_1 $ is not a constant at $ 2k\pi $ for arbitrary integer $ k $ as $ q $ changes.
Hence, we have 
\begin{equation}
	\begin{aligned}
		\hat{f} \left( r_{0}, \theta_{0}; \theta \right) = &\frac{1}{Q}\!\! \left| \sum_{q\in \mathcal{Q}}\!\exp{\l(\jmath\pi  (A_1 q+ A_2 q^2)\r)} \right|\\
		= &\frac{1}{Q}\!\! \left| \sum_{q\in \mathcal{Q}}\!\exp{\l(\jmath\pi  A_2(q+ \frac{A_1}{2A_2})^2\r)} \right| \label{eq:beam pattern summation}
	\end{aligned}
\end{equation}
where $ A_1 = U\Delta $ and $ A_2 = \frac{(Ud_0)^2}{\lambda}\frac{1-\theta_{0}^2}{r_{0}} $.

Then, the summation in \eqref{eq:beam pattern summation} can be approximated by an integral, which is given by
\begin{equation}
	\label{eq:Fresnel}
	\begin{aligned}
		&\hat{f} \left( r_{0}, \theta_{0}; \theta \right) \overset{{(b_1)}}\approx \frac{1}{Q}\!\! \left| \int_{-\frac{Q}{2}}^{\frac{Q}{2}} \!\exp{\l(\jmath\pi  (A_1 q+ A_2 q^2)\mathrm{d}q \r)} \right|\\
		&~~~~~~\overset{{(b_2)}}= \frac{1}{Q}\!\! \left| \frac{1}{\sqrt{2A_2}}\int_{\sqrt{2A_2}(-\frac{Q}{2}+\frac{A_1}{2A_2})}^{\sqrt{2A_2}(\frac{Q}{2}+\frac{A_1}{2A_2})}\!e^{( \frac{\jmath\pi t^2}{2})}\mathrm{d}t \right|\\
		&=\!\! \left|\!\!  \frac{\int_{0}^{\sqrt{2A_2}(\frac{Q}{2}+\frac{A_1}{2A_2})}\!\!e^{( \frac{\jmath\pi t^2}{2})}\!\mathrm{d}t\!\! -\!\! \int_{0}^{\sqrt{2A_2}(-\frac{Q}{2}+\frac{A_1}{2A_2})}\!\!e^{( \frac{\jmath\pi t^2}{2})}\mathrm{d}t}{\sqrt{2A_2}Q}\! \right|,
	\end{aligned}
\end{equation}
where $ (b_1) $ is due to the approximation from the summation to the integral and $ (b_2) $ is obtained by setting $ A_2(q+\frac{A_1}{2A_2})^2 = \frac{t^2}{2} $.
Let $ \beta_{1} = \frac{A_1}{2A_2} = \Delta\sqrt{\frac{r_{0}}{d_0(1-\theta_{0}^2)}}$ and $ \beta_{2} = \frac{\sqrt{A_2}Q}{2} = \frac{QU}{2}\sqrt{\frac{d_0(1-\theta_{0}^2)}{r_{0}}}$, \eqref{eq:Fresnel} can be simplified as
\begin{equation}
	\begin{aligned}
		\hat{f} \left( r_{0}, \theta_{0}; \theta \right) & = \left| \frac{\int_{0}^{\beta_1+\beta_2}e^{( \frac{\jmath\pi t^2}{2})}\mathrm{d}t - \int_{0}^{\beta_1-\beta_2}e^{( \frac{\jmath\pi t^2}{2})}\mathrm{d}t}{2\beta_2} \right|\\
		&=G(\beta_{1},  \beta_{2}),
	\end{aligned}
\end{equation}
where $G(\beta_{1},  \beta_{2}) \triangleq (\widehat{C}(\beta_{1},\beta_{2}) +  \jmath(\widehat{S}(\beta_{1},\beta_{2}))/(2\beta_2)$,
$ \widehat{C}(\beta_{1},\beta_{2}) \triangleq {C}(\beta_{1}+\beta_{2}) - C(\beta_{1}-\beta_{2})$ {\rm and} $ \widehat{S}(\beta_{1},\beta_{2}) \triangleq S(\beta_{1}+\beta_{2}) - S(\beta_{1}-\beta_{2}) $. Specifically, $ C(x) = \int_{0}^{x} \cos(\frac{\pi}{2}t^2 ){\rm d}t $ {\rm and} $ S(x) = \int_{0}^{x} \sin(\frac{\pi}{2}t^2 ){\rm d}t $ {\rm are the Fresnel integrals.
The proof of Lemma \ref{lemma:LSAgeneral} . 

%\bibliography{IEEEabrv,ref/bibfile}
\bibliographystyle{IEEEtran}
\bibliography{IEEEabrv}

\end{document}

%% file: header.tex
\newtheorem{definition}{\emph{\underline{Definition}}}

\newtheorem{lemma}{\emph{\underline{Lemma}}}
\newtheorem{observation}{\emph{\underline{Observation}}}

\newtheorem{proposition}{\emph{\underline{Proposition}}}

\newtheorem{remark}{\bf \emph{\underline{Remark}}}

\def\l{\left}
\def\r{\right}
\def\({\left(}
\def\){\right)}

\def\b0{{\mathbf{0}}}

\newcommand{\nn}{\nonumber}